\documentclass{article}
\usepackage{subcaption}
\usepackage{graphicx} 
\usepackage{amsthm,amsmath,amsfonts,amssymb,graphicx,bm,color}
\usepackage[top=2cm,bottom=2cm,left=2cm,right=2cm]{geometry}
\usepackage{url}
\usepackage{array}
\usepackage{arydshln} 
\usepackage{algorithm}
\usepackage{algorithmic}

\theoremstyle{definition}
\newtheorem{definition}{Definition}
\newtheorem{lemma}{Lemma}

\newtheorem{corollary}{Corollary}
\newtheorem{example}{Example}
\newtheorem{remark}{Remark}
\newtheorem{proposition}{Proposition}

\usepackage{authblk}
\title{Minimum information Markov model}
\author[1]{Issey Sukeda}
\author[2]{Tomonari Sei}
\affil[1]{The University of Tokyo, RIKEN Center for Brain Science}
\affil[2]{The University of Tokyo}
\date{}

\begin{document}
\maketitle
\begin{abstract}
The analysis of high-dimensional time series data has become increasingly important across a wide range of fields. Recently, a method for constructing the minimum information Markov kernel on finite state spaces was established. In this study, we propose a statistical model based on a parametrization of its dependence function, which we call the \textit{Minimum Information Markov Model}.
We show that its parametrization induces an orthogonal structure between the stationary distribution and the dependence function, and that the model arises as the optimal solution to a divergence rate minimization problem.
In particular, for the Gaussian autoregressive case, we establish the existence of the optimal solution to this minimization problem, a nontrivial result requiring a rigorous proof.
For parameter estimation, our approach exploits the conditional independence structure inherent in the model, which is supported by the orthogonality. Specifically, we develop several estimators, including conditional likelihood and pseudo likelihood estimators, for the minimum information Markov model in both univariate and multivariate settings. We demonstrate their practical performance through simulation studies and applications to real-world time series data.
\end{abstract}

\section{Introduction}

The analysis of high-dimensional time series data has become increasingly important across a variety of fields, including finance, neuroscience, genomics, and environmental science. In such applications, observations are not only multivariate but also temporally dependent, posing significant challenges for statistical modeling and inference. Traditional time series models, such as vector autoregressive (VAR) models, often require strong assumptions on its structure to remain computationally and statistically feasible. 
Other approaches, such as mixture transition distribution model~\cite{berchtold2002mixture}, employ higher-order Markov chains which is typically described by a parametric family of Markov kernels. 

To address these challenges, some approaches have focused on specifying the stationary distribution of the given time series first. If a parametric model of marginal stationary distributions is given, the estimation of Markov kernel is relatively easy because the methods for independent data can be applied under ergodicity.
Similar approaches are relatively common in dependence modeling for high-dimensional i.i.d. data. One typical example is the copula modeling, where the joint distribution is decomposed into the corresponding marginal distributions and the copula function that combines these marginals. 
Contrarily, such ideas are not popular in time series analysis. Although studies on copula-based time series modeling has been growing~\cite{patton2012review,sun2020copula}, these are quite different because the focus is to model the transition kernel using copula modeling. 

In this paper, we propose a new framework for modeling high-dimensional time series using higher-order Markov kernel called the \textit{minimum information Markov kernel}. 
This Markov kernel is recently established by Sei~\cite{sei2024constructingmarkovchainsgiven} for finite state spaces based on the information geometry of the exponential family of Markov chains~\cite{hayashi2016information,nagaoka2005exponentialfamilymarkovchains}, and is characterized by two components: the stationary distribution and the dependence function. 
The key feature of this kernel is that it enables the modeling of dependence structure and the marginal distribution separately. This approach is common in statistical modeling for joint distributions under i.i.d. settings, especially known by the copula modeling~\cite{nelsen2006introduction}. 
To incorporate its advantages to time series domain, our idea is to assign parameters linearly on its dependence function of the minimum information Markov kernel and use it as a statistical model, which we call the \textit{Minimum information Markov model} hereafter. Our approach leverages conditional independence structures within this kernel to reduce model complexity while retaining flexibility. Specifically, we develop several estimators for our minimum information Markov model in a high-dimensional regime and demonstrate their practical performance through simulation experiments and real data applications.

The name \textit{Minimum information Markov model} originates from the fact that the model can be interpreted as a minimum information framework, as introduced in previous studies. Here, ``information'' refers to negative Shannon entropy in the information-theoretic sense.
The concept of a minimum information (or minimally informative) copula was first proposed by Meeuwissen and Bedford~\cite{meeuwissen1997minimally}, who derived it as the optimal solution to an information minimization problem subject to a constraint fixing Spearman's $\rho$.
A number of studies have investigated its properties, characterizations~\cite{sukeda2025minimum,sukeda2025frank}, estimations~\cite{chen2024proper}, applications~\cite{shekari2023maximum}, and extensions~\cite{pougaza2012new}. One of the most general extensions is the minimum information dependence model~\cite{sei2024minimum}, which extends the copula to higher dimensions, arbitrary domains, and arbitrary marginal distributions. These prior works focus exclusively on i.i.d. data, whereas our work can be seen as a straightforward extension to time series data and Markov models.

The remainder of the paper is organized as follows.
Section 2 introduces the minimum information Markov model and shows that its parametrization induces an orthogonal structure between the stationary distribution and the dependence function, and that the model arises as the optimal solution to a divergence rate minimization problem. In particular, for the Gaussian autoregressive case, we further establish the existence of an optimal solution to this minimization problem, a nontrivial result requiring a rigorous proof. Other cases remain an open problem.
Section 3 presents estimation procedures for the proposed model. 
Section 4 reports simulation studies and empirical analyses. 
Section 5 applies the proposed method to real neural data. 
Section 6 concludes with a discussion of possible extensions. 
Proofs are primarily deferred to the Appendix.

\subsection{Preliminaries and notation}
Throughout the paper, $I$ denotes the identity matrix with appropriate size. The matrix size will not be displayed explicitly. 
$\otimes$ denotes the Kronecker product. 
To avoid confusion, $\mathrm{vec}(\cdot)$ denotes the vectorization operator that stacks the column vectors, not the row vectors.
The compatibility between the Kronecker product and the vectorization operator will be primarily used: $\mathrm{vec}(ABC) = (C^\top \otimes A) \mathrm{vec}(B)$ for three square matrices $A, B, C$. 

The Gaussian integral for $p$-dimensional variate $x$ is defined as 
$$\int \exp{\left(-\frac{1}{2}x^\top A x + b^\top x\right)}dx = \sqrt{\frac{(2\pi)^p}{\mathrm{det}(A)}} \exp{(\frac{1}{2}b^\top A^{-1} b)}$$
$$\int( x^\top R x)\exp{\left(-\frac{1}{2}x^\top A x\right)}dx = \sqrt{\frac{(2\pi)^p}{\mathrm{det}(A)}} \mathrm{Tr}(RA^{-1})$$
where $\mathrm{Tr}(\cdot)$ denotes the trace of matrix, $A$ and $R$ are positive definite $p\times p$ matrices, and $b$ is a $p$-dimensional vector.

We abuse notation by writing $\mathcal{N}(x;0,B)$ to denote a Gaussian distribution with zero mean and covariance matrix $B$ when the variable $x$ is explicitly shown; we may also write $\mathcal{N}(0,B)$ when the variable is omitted.

For two square matrices $A$ and $B$, the notation $A \succ B$ indicates that $A - B$ is positive definite, while $A \succeq B$ indicates that $A - B$ is positive semidefinite. Unless stated otherwise, $O$ denotes the zero matrix of appropriate size. Accordingly, the positive definiteness of a matrix $A$ is expressed as $A \succ O$.

\section{Minimum information Markov model}

Recently, Sei~\cite{sei2024constructingmarkovchainsgiven} proposed the construction of Markov chains on a finite state space with given dependence and marginal stationary distribution, based on the exponential family of Markov chains~\cite{nagaoka2005exponentialfamilymarkovchains,hayashi2016information}. Its Markov kernel has been named the minimum information Markov kernel, since it is obtained by KL-like divergence rate minimization problems. 

Let $\mathcal{X}$ denote the state space. 
Let $\mathbb{R}_+$ and $\mathbb{R}_{\geq 0}$ be the set of positive and non-negative numbers, respectively.
The set of all positive probability distributions on $\mathcal{X}$ is denoted by $P_+(\mathcal{X})$. 
In what follows, we focus on stationary Markov kernels throughout the paper. 
We formally extend the definition of the minimum information Markov kernel for general state spaces as follows.

\begin{definition}[Minimum information Markov kernel (first order)]
    Given a function $H:\mathcal{X}^2 \to \mathbb{R}$ and the marginal distribution $r \in P_+(\mathcal{X})$, the Markov kernel
    \begin{equation} \label{eq:mininfo-markov-kernel}
    w(y|x) = \exp{(H(x,y) + \kappa(y) - \kappa(x) - \delta(y))}
    \end{equation}
    with the stationary distribution $p_w(x) = r(x)$ is called the minimum information Markov kernel when the normalizing functions $\kappa:\mathcal{X}\to \mathbb{R}$ and $\delta:\mathcal{X}\to \mathbb{R}$ both exist. 
\end{definition}

In Sei~\cite{sei2024constructingmarkovchainsgiven}, the theoretical guarantee of its existence and the framework of specifying the dependence function $h(x,y)$ and the stationary distribution $r(x)$ are mainly presented\footnote{Note that Sei~\cite{sei2024constructingmarkovchainsgiven} assumes a finite state space throughout the paper, while the extension of the theory to an infinite state space is only mentioned slightly in Section 5.1, giving a conjecture that the construction is still valid for infinite spaces. Despite of it, we shall apply the minimum information model to the case where the state space is infinite-dimensional such as $\mathcal{X} = \mathbb{R}$. In fact, in AR/VAR models, the corresponding Markov kernel exists thus the argument is valid.}, while the statistical application is limited to a mention in Section 5.2 of their paper.
In this study, we explore a method to utilize this Markov kernel to model $p$-variate time series.
For statistical modeling, it is convenient to consider a linearly parametrized dependence function, hence we replace $H(x,y)$ in \eqref{eq:mininfo-markov-kernel} with $\bm{\theta}^\top \bm{h}(x,y)$. Each $\bm{h}_i$ for $1 \leq i \leq K\ (\in \mathbb{N})$ denotes different type of dependencies and the parameter $\bm{\theta}$ indicates their weights.

\begin{definition}[$d$-th-order minimum information Markov model] \label{def:mininfomarkov}

Let $\mathcal{X} = \mathbb{R}^p$. A sequence $(x_s, x_{s-1}, \dots, x_t)$ is abbreviated as $x_{s:t}$ for $s>t$. Given the observed data $x_{n:1}$ and the dependence function $\bm{h}:\mathcal{X}^{d+1} \to \mathbb{R}^K$ , the $d$-th-order minimum information Markov model is defined as 
\begin{equation} \label{eq:mininfo-markov-model}
p(x_t|x_{t-1:t-d}) = \exp{(\bm{\theta}^\top \bm{h}(x_{t:t-d}) + \kappa(x_{t:t-d+1}) - \kappa(x_{t-1:t-d}) - \delta(x_t))}.
\end{equation}
\end{definition}
\noindent This model is specified by the dependence function $\bm{h}$ and the $p$-dimensional stationary distribution, i.e., the marginal distribution over states at time $t$. The joint distribution is then written as 
\begin{equation}
p(x_{n:1}) = \exp{\left(\sum_{t=d+1}^n  \bm{\theta}^\top \bm{h}(x_{t:t-d}) + \kappa(x_{n:n-d+1}) - \kappa(x_{d:1}) - \sum_{t=d+1}^n 
 \delta(x_t) \right)}. \label{eq:joint}
\end{equation}
Here we point out that the conditional inference is possible due to the form of likelihood function. Specifically, consider a partial permutation on the observed data $x_{n:1}$ where $x_{d:1}$ and $x_{n:n-d+1}$ are fixed and the rest is permuted randomly. It is clear that the terms with intractable regularization functions $\kappa$ and $\delta$ in Equation~\eqref{eq:joint} are invariant to such permutations. Therefore, maximum conditional likelihood estimation (CLE) is applicable by conditioning on such a permutation, which will be formally presented in Section 3.

Conceptually, the minimum information Markov model is characterized by its dependence function and stationary distribution. A natural question that arises is: Does such a kernel exist uniquely? In this section, we first show that, for a given dependence function and stationary distribution, the minimum information Markov model --- if it exists --- is unique. This result follows from the orthogonal structure between the two, as viewed through the lens of information geometry. We then demonstrate that, in the case of Gaussian autoregressive models, which are notable examples of the minimum information Markov model, existence is also guaranteed. Furthermore, we point out that this unique existence yields a new unconstrained parametrization for Gaussian autoregressive models, distinct from the classical parametrization. In general, however, existence is not guaranteed, although we can formally consider several examples included in the minimum information Markov model.

\subsection{Uniqueness and Orthogonality}

\subsubsection{Pythagorean structure}
Information geometric concepts are useful to illustrate the minimum information Markov model.
On a finite state space setting, divergence has been defined for transition matrices of discrete-time Markov chains~\cite{wang2023information, wolfer2023information,sei2024constructingmarkovchainsgiven}. The following definition is an extension of it to an infinite state space setting.
\begin{definition}[Divergence rate]
For two $d$-th order stationary Markov kernels $p$ and $q$, the divergence rate is defined as 
$$D(p|q) = \int p(x_{t-1:t-d}) \left(\int p(x_{t}|x_{t-1:t-d})\log\frac{p(x_{t}|x_{t-1:t-d})}{q(x_t|x_{t-1:t-d})} dx_{t} \right) dx_{t-1:t-d},$$    
where $p(x_{t-1:t-d})$ denotes the joint distribution.
\end{definition}

Using this divergence rate, the following orthogonality holds, leading to a separate view of the dependence function and the stationary distribution given by the minimum information Markov kernel, which is illustrated in Figure~\ref{fig:space_ar1}.
\begin{proposition}[Generalized Pythagorean Theorem~\cite{csiszar1987conditional}] \label{prop:pythagorean}
Let $\mathcal{W}$ denote the set of 1st order stationary Markov kernels defined on $\mathbb{R}^p$.
Let $\mathcal{E}$ be a family of Markov kernels with dependence structure $\bm{\theta}^\top \bm{h}(x_t,x_{t-1})$;
$$\mathcal{E} = \{w(x_t|x_{t-1}) = \exp{(\bm{\theta}^\top \bm{h}(x_t,x_{t-1}) + \kappa(x_t)-\kappa(x_{t-1}) - \delta(x_t))}\ |\ w \in \mathcal{W}, \kappa, \delta: \mathbb{R}^p \to \mathbb{R}\}$$
Let $\mathcal{M}$ be the set of all Markov kernels $w$ satisfying
$$\int w(x_t|x_{t-1}) r(x_{t-1}) dx_{t-1} = r(x_t)$$
for given $r$. 
If there exists a unique $w_* \in \mathcal{M} \cap \mathcal{E}$, then
$$D(w|w_*) + D(w_*|v) = D(w|v), w \in \mathcal{M}, v \in \mathcal{E}.$$
\end{proposition}
\begin{proof}
    \begin{align*}
        &D(w|w_*) + D(w_*|v) - D(w|v) \\
        &= \int \int \left( w(x_t,x_{t-1}) \log{\frac{w}{w_*}} + w_*(x_t,x_{t-1}) \log{\frac{w_*}{v}} - w(x_t,x_{t-1}) \log{\frac{w}{v}}\right) dx_tdx_{t-1}\\
        &= \int \int \left( w(x_t,x_{t-1}) \log{\frac{v}{w_*}} + w_*(x_t,x_{t-1}) \log{\frac{w_*}{v}}\right) dx_tdx_{t-1}\\
        &= \int \int \left( \{w(x_t,x_{t-1}) - w_*(x_t,x_{t-1}) \}\log{\frac{v}{w_*}}\right) dx_tdx_{t-1}\\
        &= \int \int  \{w(x_t,x_{t-1}) - w_*(x_t,x_{t-1}) \}((\kappa_v(x_t)-\kappa_v(x_{t-1}) - \delta_v(x_t)) - (\kappa_{w_*}(x_t)- \kappa_{w_*}(x_{t-1}) - \delta_{w_*}(x_t))) dx_tdx_{t-1}\\
        &=0
    \end{align*}
    The last equality follows from the fact that $w$ and $w_*$ shares the same stationary distribution $\nu$. 
    The uniqueness holds because if $w \in \mathcal{M} \cap \mathcal{E}$, then $D(w|w_*) + D(w_*|w) = 0$ holds by taking $v=w$, which implies $w_* = w$.
\end{proof}
\noindent Not only the orthogonality, this proposition states that $w_*$ is the optimal solution, if exists, of the divergence rate minimization problem 
\begin{equation} \label{eq:divergence-rate-minimization}
\mathrm{minimize}_{w \in \mathcal{M}}\ D(w|v), v\in \mathcal{E}.
\end{equation}
The existence of an optimal solution is not guaranteed in general. However, it does hold in cases involving Gaussian autoregressive models, which are notable examples of the minimum information Markov model.

\begin{figure}[t]
    \centering
    \includegraphics[width=0.5\linewidth]{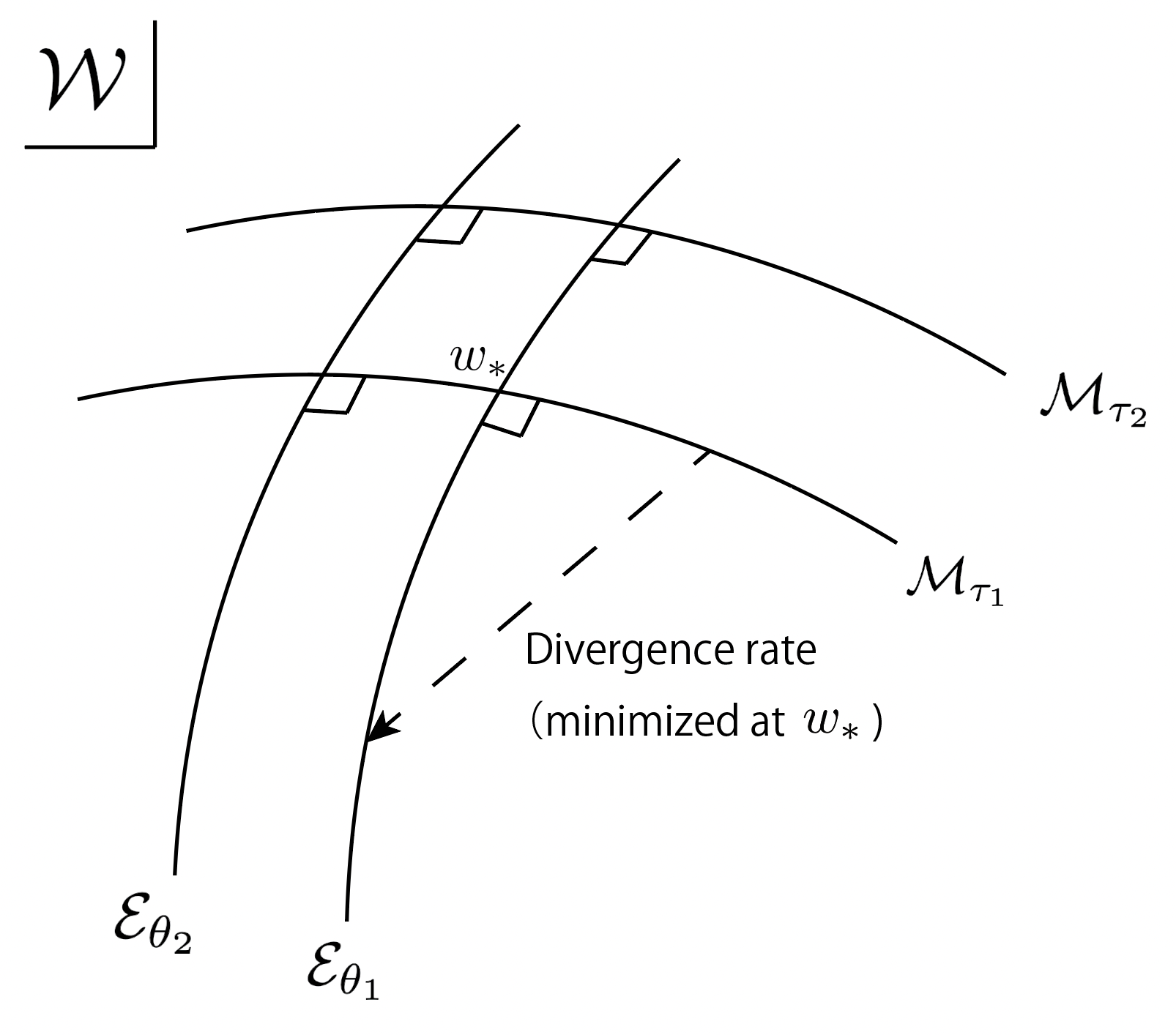}
    \caption{Illustration of Proposition~\ref{prop:pythagorean}.}
    \label{fig:space_ar1}
\end{figure}

\subsubsection{Fisher information}
Proposition~\ref{prop:pythagorean} displays the global orthogonal structure of the space of minimum information Markov models. Contrarily, to explore the local structure of the space of high-dimensional probability distribution, considering the Fisher information matrix is a standard approach. For the minimum information Markov kernel $p$ with the stationary distribution $r$, the Fisher information matrix $G(\xi=\{\theta,r\})$ is calculated as 
$$G(\xi) = -\mathbb{E}[\nabla^2 \log{p_{\xi}(x_t|x_{t-1:t-d})}] = -\mathbb{E}[\nabla^2 (\kappa_{\xi}(x_{t:t-d+1})- \kappa_{\xi}(x_{t-1:t-d}) -\delta_{\xi}(x_t))] = \mathbb{E}[\nabla^2\delta_{\xi}(x_t)].$$
The last equality follows from the stationarity of the process. This is similar to the Fisher information for the exponential family of finite-state first-order Markov chains studied in Nagaoka~\cite{nagaoka2005exponentialfamilymarkovchains}, 
$$\log{p_\theta(x_t|x_{t-1})} = C(x_t,x_{t-1}) + \bm{\theta}^\top \bm{h}(x_t,x_{t-1}) + \kappa_\theta(x_t) - \kappa_\theta(x_{t-1}) - \psi(\theta);$$
$$G(\theta) = \nabla_\theta^2 \psi(\theta),$$ 
but differs in that $\delta$ is a function of $x_t$. 

Furthermore, we show the orthogonality in the case $d=1$ as an example.
\begin{proposition}[Fisher information of the first-order minimum information Markov model] \label{prop:fim}
Assume the differentiation with respect to the parameters and integration over the states can be interchanged.
The Fisher information of the first-order minimum information Markov model
$$w(y|x) = \exp{(\theta h(x,y) + \kappa(y)-\kappa(x)-\delta(y)})$$
$$p_w(x) = r(x; \nu)$$
is
\[
\begin{cases}
    g_{\theta\theta} &= \int r(y) \frac{\partial^2}{\partial\theta^2}\delta(y) dy = \frac{\partial^2}{\partial\theta^2}\int r(y) \delta(y) dy, \\
    g_{\theta\nu} &= \int r(y) \frac{\partial^2}{\partial\theta\partial\nu}\delta(y) dy = 0,\\
    g_{\nu\nu} &= \int r(y) \frac{\partial^2}{\partial\nu^2}\delta(y) dy = -\int \frac{\partial}{\partial\nu}r(y) \frac{\partial}{\partial\nu}\delta(y) dy.
\end{cases}
\]
Especially, the parameters $\theta$ and $\nu$ are orthogonal.
\end{proposition}

\subsection{Gaussian autoregressive models}

As notable examples, we discuss the autoregressive (AR) model and the vector autoregressive (VAR) model. Specifically, we show that when the dependence function has a multiplicative form and the stationary distribution is Gaussian, the minimum information Markov model exists uniquely and is equivalent to an AR model. A similar result holds for VAR models.

For the identifiability of the minimum information Markov model, we first define the following equivalence relationship.

\begin{definition} \label{def:equivalence}
Suppose $d\geq 1$. Let $\mathcal{L}$ be the linear space consisting of all functions that can be written in the form:
\[
\kappa(x_{d:2}) - \kappa(x_{d-1:1}) - \delta(x_d)
\]
where \( \kappa \) and \( \delta \) are arbitrary functions (with appropriate domains).
Then, $f \sim g$ if and only $f-g \in \mathcal{L}$.
\end{definition}

Under this notion, it is clear that Gaussian AR/VAR models are included in the minimum information Markov model. In other words, the classical parameterization of AR/VAR models, given by the coefficients and noise variance, can be replaced by the parameterization of the minimum information Markov model, given by the dependence parameter and the variance of the stationary distribution.
\begin{example}[AR(1) model] \label{example:ar1}
 Consider $\mathcal{X}=\mathbb{R}$.
 Define a Markov kernel $w:\mathbb{R}^2\to\mathbb{R}_+$ by
 \[
  w(y|x) = \frac{1}{\sqrt{2\pi\sigma^2}}\exp\left(-\frac{(y-\phi x)^2}{2\sigma^2}\right),\quad (x,y)\in\mathbb{R}^2,
 \]
 that satisfies $\int_{-\infty}^\infty w(y|x)dy=1$,
 where $\phi\in\mathbb{R}$ is called the AR coefficient and $\sigma^2>0$ is the noise variance.
 It is widely known that there exists a unique stationary distribution if and only if $\phi\in(-1,1)$.
 The model is represented in the form of a minimum information Markov kernel as
 \begin{align*}
  w(y|x)
  &= \exp\left(\theta xy + \kappa(y) - \kappa(x) - \delta(y)
  \right),
  \\
  p_w(x)
  &= (2\pi\tau^2)^{-1/2}\exp\left(-\frac{x^2}{2\tau^2}
  \right),
 \end{align*}
 where $\theta=\phi/\sigma^2$, $\kappa(y)=\phi^2y^2/(2\sigma^2)$, $\delta(y)=(1+\phi^2)y^2/(2\sigma^2)+(1/2)\log(2\pi\sigma^2)$ and $\tau^2=\sigma^2/(1-\phi^2)$.
 We have an additional parameter $\tau^2>0$ in the stationary distribution.
 It is shown that the parameter $(\phi,\sigma^2)\in(-1,1)\times\mathbb{R}_+$ corresponds one-to-one with
 the new parameter $(\theta,\tau^2)\in\mathbb{R}\times\mathbb{R}_+$.
 Indeed, the inverse transform is expressed as 
 \[
  \sigma^2 = \frac{2\tau^2}{1+\sqrt{1+4\theta^2\tau^4}},
  \quad \phi = \frac{2\theta\tau^2}{1+\sqrt{1+4\theta^2\tau^4}}.
 \]

Based on Proposition~\ref{prop:fim}, the Fisher information matrix $G(\xi) = \begin{pmatrix}
    g_{\theta}&g_{\theta,\tau^2}\\
    g_{\theta,\tau^2} & g_{\tau^2}
\end{pmatrix}$ becomes
\begin{equation}\label{eq:fim-ar1}
\begin{cases}
g_{\theta}(\xi) &=  \mathbb{E}[\frac{\partial^2}{\partial\theta^2}\delta(x)] = \frac{2\tau^4}{\sqrt{1+4\theta^2\tau^4} (1+\sqrt{1+4\theta^2\tau^4})}\\
g_{\tau^2}(\xi) &= \mathbb{E}[\frac{\partial^2}{\partial(\tau^2)^2}\delta(x)] = \frac{1}{2\tau^4\sqrt{1+4\theta^2\tau^4}}\\
g_{\theta,\tau^2}(\xi) &= \mathbb{E}[\frac{\partial^2}{\partial\theta \partial \tau^2}\delta(x)] = 0
\end{cases}
\end{equation}
which is depicted in Figure~\ref{fig:fim-ar1}. 

\begin{figure}[t]
    \centering
    \includegraphics[width=0.48\linewidth]{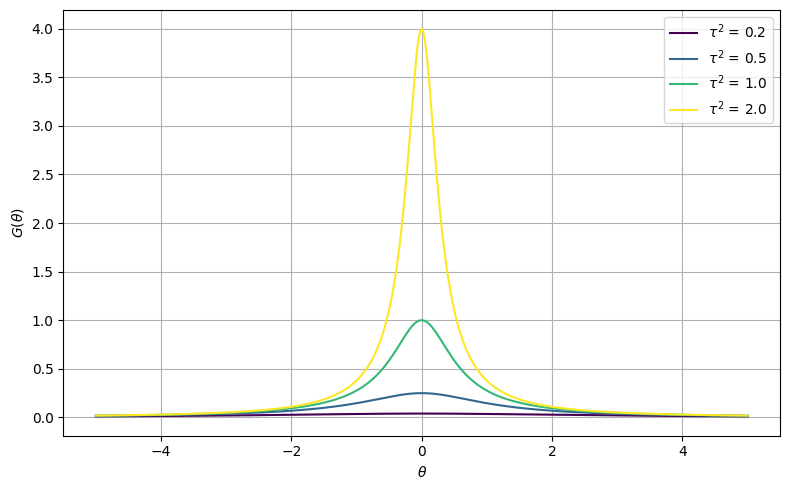}
    \includegraphics[width=0.48\linewidth]{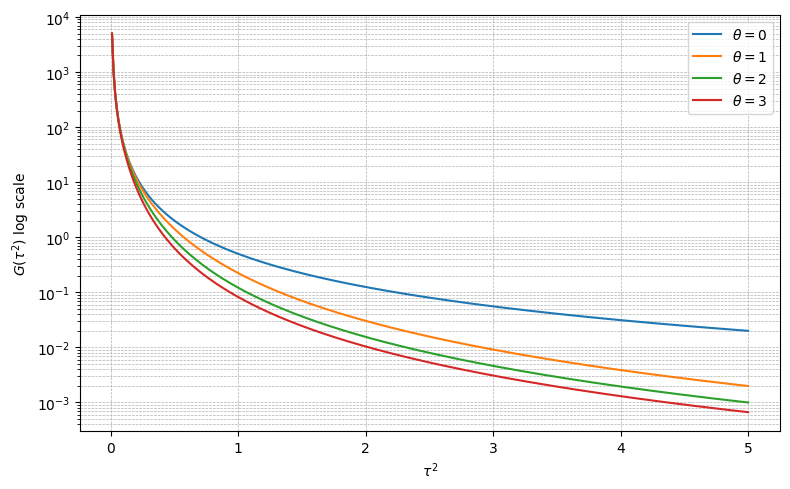}
    \caption{Fisher information of the minimum information model equivalent to the Gaussian AR(1) model (Eq.~\eqref{eq:fim-ar1}).}
    \label{fig:fim-ar1}
\end{figure}

\begin{remark}
    This Markov kernel is also represented as an exponential family studied by Nagaoka~\cite{nagaoka2005exponentialfamilymarkovchains}:
\[
w_\theta(y|x) = \exp\left(\theta_1xy + \theta_2y^2 + \kappa_\theta(y)-\kappa_\theta(x)-\psi_\theta\right),
\]
where $\theta_1=\phi/\sigma^2$, $\theta_2=-(1+\phi^2)/(2\sigma^2)$, $\kappa_\theta(y)=\phi^2 y^2/(2\sigma^2)$ and $\psi_\theta=(1/2)\log(2\pi\sigma^2)$.
Their inverse transforms are given by $\phi = \frac{-\theta_2-\sqrt{\theta_2^2-\theta_1^2}}{\theta_1}, \sigma^2 = \frac{-\theta_2-\sqrt{\theta_2^2-\theta_1^2}}{\theta_1^2}$.
However, the domain of the natural parameter $\theta=(\theta_1,\theta_2)$ is not the whole space:
\begin{align*}
\Theta = \{(\theta_1,\theta_2)\in\mathbb{R}^2\mid |\theta_1|<|\theta_2|, \theta_2<0\}
\end{align*}
due to the restrictions $\sigma^2>0$ and $\phi\in(-1,1)$.
\end{remark}
\end{example}

Similar properties do hold for higher-order and higher-dimensional AR models.

\begin{example}[AR(2) model] \label{example:ar2}
 As in the preceding example, consider $\mathcal{X}=\mathbb{R}$ and define a Markov kernel $w:\mathbb{R}^3\to\mathbb{R}_+$ by
 \[
  w(x_3|x_1,x_2) = (2\pi\sigma^2)^{-1/2}\exp\left(-\frac{1}{2\sigma^2}(x_3-\phi_1 x_2-\phi_2x_1)^2\right),
 \]
 where $(\phi_1,\phi_2)\in\mathbb{R}^2$ and $\sigma^2>0$.
 There exists a unique stationary distribution of $w$ if and only if the two roots of the quadratic equation
 $1-\phi_1 z - \phi_2 z^2 = 0$ are outside the unit circle.
 It is known that there exists a unique stationary distribution of $w$
 if and only if $(\phi_1,\phi_2)\in S$, where
 \[
  S = \left\{
  (\phi_1,\phi_2)\mid 1+\phi_2>0,\ 1-\phi_1-\phi_2>0,\ 1+\phi_1-\phi_2>0
  \right\}
 \]
(see e.g., Marmol~\cite{marmol1995stationarity}).

Based on Definition~\ref{def:equivalence}, the AR(2) model is represented in the form of a minimum information Markov kernel as
 \begin{align*}
 w(x_3|x_1,x_2) &= \exp\left(\theta_1 x_2x_3+ \theta_2x_1x_3 + \kappa(x_2,x_3)-\kappa(x_1,x_2)-\delta(x_3)\right),
 \\
 p_w(x) &= (2\pi\tau^2)^{-1/2}\exp\left(-\frac{x^2}{2\tau^2}\right).
 \end{align*}
where 
\begin{align*}
 \kappa(x_2,x_3) &= \frac{(\phi_1^2+\phi_2^2)x_3^2+2\phi_1\phi_2x_2x_3+\phi_2^2x_2^2}{2\sigma^2},\\
\delta(x_3) &= \frac{1}{2\sigma^2}(1+\phi_1^2+\phi_2^2)x_3^2 + \frac{1}{2}\log(2\pi\sigma^2),\\
\theta_1&=\phi_1(1-\phi_2)/\sigma^2,\\
\theta_2&=\phi_2/\sigma^2,\\
\tau^2 &= \frac{(1-\phi_2)\sigma^2}{(1+\phi_2)(1-\phi_1-\phi_2)(1+\phi_1-\phi_2)}.
\end{align*}
It is shown that the parameter $(\phi_1,\phi_2,\sigma^2)\in S\times\mathbb{R}_+$ corresponds one-to-one with $(\theta_1,\theta_2,\tau^2)\in\mathbb{R}^2\times\mathbb{R}_+$. 
Clearly, $(\theta_1,\theta_2,\tau^2)$ is uniquely obtained from $(\phi_1,\phi_2,\sigma^2)$, however the opposite is not trivial. 

First, $\phi_1$ and $\phi_2$ can be expressed as 
$$
\begin{cases}
  \phi_1 = \frac{\theta_1 \sigma^2}{1-\theta_2\sigma^2}  \\
  \phi_2 = \theta_2 \sigma^2
\end{cases}
$$
and then $\tau^2$ is expressed as the function of $\sigma^2$ as 
\begin{align*}
    \tau^2(\sigma^2) &= \frac{(1-\theta_2 \sigma^2)\sigma^2}{(1+\theta_2 \sigma^2)(1-\frac{\theta_1 \sigma^2}{1-\theta_2\sigma^2} - \theta_2 \sigma^2)(1 + \frac{\theta_1 \sigma^2}{1-\theta_2\sigma^2}-\theta_2\sigma^2)} \\
    &= \frac{(1-\theta_2 \sigma^2)^3 \sigma^2}{(1+\theta_2\sigma^2)\{(1-\theta_2\sigma^2)^4 - (\theta_1 \sigma^2)^2 \}}
\end{align*}
\noindent Here, we show the one-to-one correspondence between $(\theta_1, \theta_2, \tau^2)$ and $(\phi_1, \phi_2, \sigma^2)$ via showing the monotonicity of $\tau^2$ with respect to $\sigma^2$. In doing so, we denote $\sigma^2 = t$ and show 
$$g(t) = \frac{1}{\tau^2(t)} = \frac{1-\theta_2^2t^2}{t} - \frac{(1+\theta_2 t)\theta_1^2 t}{(1-\theta_2 t)^3}$$
is monotone with respect to $t$ such that $(\phi_1, \phi_2) \in S$. The derivative of $g$ is 
\begin{align*}
    g'(t) &= - \frac{1+\theta_2^2t^2}{t^2} - \frac{(1+4\theta_2t + \theta_2^2 t^2)\theta_1^2}{(1-\theta_2t)^4}\\
    &=- \frac{(1+\theta_2^2t^2)(1-\theta_2t)^4 + (1+4\theta_2t + \theta_2^2t^2)\theta_1^2t^2}{t^2(1-\theta_2t)^4}\\
    &< - \frac{(1+\theta_2^2t^2)\theta_1^2t^2 + (1+4\theta_2t + \theta_2^2t^2)\theta_1^2t^2}{t^2(1-\theta_2t)^4}\\
    &=  - \frac{2\theta_1^2t^2(1+\theta_2t)^2}{t^2(1-\theta_2t)^4}\\
    &< 0.
\end{align*}
\noindent With the monotonicity of $g$, $\tau^2(t) \to 0$ when $t \to 0$, and $\tau^2(t) \to \infty$ when $(\theta_1, \theta_2, S)$ goes to the boundary of $S$, $\tau^2$ and $t = \sigma^2$ are one-to-one.
It is clear that $(\phi_1,\phi_2)=(\frac{\theta_1t}{1-\theta_2t},\theta_2 t)$ approaches the boundary of $S$ as $t$ grows by the following argument. 

Define the margins to the three sides of $S$ by
\[
b_1(t)= 1+\phi_2(t),\qquad
b_2(t)= 1-\phi_1(t)-\phi_2(t),\qquad
b_3(t)=1+\phi_1(t) - \phi_2(t).
\]

\emph{Case 1: $\theta_2<0$.}
Then $b_1(t)=1+\theta_2 t$ is linear and vanishes at $t_*=-1/\theta_2>0$.
Hence the path hits the side $1+\phi_2=0$.

\emph{Case 2: $\theta_2=0$.}
Here $\phi_2 = 0$ and $\phi_1(t)=\theta_1 t$.
If $\theta_1>0$ then $b_2(t)=1-\theta_1 t$ vanishes at $t_*=1/\theta_1$,
so the path hits $1-\phi_1-\phi_2=0$.
If $\theta_1<0$ then $b_3(t)=1+\theta_1 t$ vanishes at $t_*=1/|\theta_1|$,
so it hits $\phi_2-\phi_1=1$.
The degenerate pair $\theta_1=\theta_2=0$ yields the constant point $(0,0)$. 

\emph{Case 3: $\theta_2>0$.}
The domain is $[0,1/\theta_2)$.
As $t\uparrow 1/\theta_2$ we have $\phi_2(t)\uparrow 1$ and
\[
\phi_1(t)=\frac{\theta_1 t}{1-\theta_2 t}\to
\begin{cases}
+\infty & \text{if }\theta_1>0,\\
-\infty & \text{if }\theta_1<0.
\end{cases}
\]
Hence $b_2(t)=1- \phi_1(t)-\phi_2(t)\to -\infty$ if $\theta_1>0$,
while $b_3(t)=1 + \phi_1(t) - \phi_2(t)\to -\infty$ if $\theta_1<0$.
Since $b_2(0)=b_3(0)=1>0$ and both $b_2$ and $b_3$ are continuous on $[0,1/\theta_2)$,
the intermediate value theorem gives $t_*\in(0,1/\theta_2)$ with either $b_2(t_*)=0$ or $b_3(t_*)=0$,
so the path hits either $1 - \phi_1 - \phi_2=0$ or $1 +\phi_1 -\phi_2=0$.
If $\theta_1=0$ then $\phi_1 = 0$ and $b_2(t)=b_3(t)=1-\theta_2 t\downarrow 0$
as $t\uparrow 1/\theta_2$, hence the path approaches the upper boundary.
\end{example}

\begin{example}[AR$(d)$ model] \label{example:ar-d}
    The AR($d$) model is usually defined in the following form:
    $$x_t = \sum_{i=1}^d \phi_i x_{t-i} + \epsilon_t, \epsilon_t \sim \mathcal{N}(0,\sigma^2), $$
    where $\epsilon_t$ is the noise term, and $\phi_i$s are the AR parameters. 
    Similar to the preceding examples, consider $\mathcal{X} = \mathbb{R}$ and define a Markov kernel $w:\mathbb{R}^{d+1} \to \mathbb{R}_+$ by 
    \begin{align*}
    w(x_t|x_{t-1:t-d}) &= (2\pi\sigma^2)^{-1/2}\exp{\left(-\frac{1}{2\sigma^2} (x_t-\sum_{i=1}^d \phi_i x_{t-i})^2 \right)}\\
    &= \exp{\left(\sum_{i=1}^d \{\frac{\phi_i - \sum_{1\leq j < k \leq  d,k-j=i} \phi_j \phi_k}{\sigma^2}\} x_t x_{t-i} + \kappa(x_{t:t-d+1})-\kappa(x_{t-1:t-d}) - \delta(x_{t}) \right)},\\
    p_w(x) &= (2\pi\tau^2)^{-1/2}\exp\left(-\frac{x^2}{2\tau^2}\right).
    \end{align*}
    where $\kappa$ and $\delta$ are certain functions. 
    The corresponding dependence function $\bm{h} : \mathcal{X}^{d+1} \to \mathbb{R}^d$ is $\bm{h}_i(x_{t:t-d}) = x_t x_{t-i}$ for $i=1,\dots, d$. 
    There exists a unique stationary distribution of $w$ if and only if the roots of the equation $1-\phi_1z-\phi_2z^2-\dots-\phi_dz^d=0$ are outside the unit circle. 
    Analogously to the previous examples, $(\phi_1,\dots, \phi_d, \sigma^2)$ is expected to correspond in a one-to-one manner to $(\theta_1, \dots, \theta_d, \tau^2)$.
\end{example}

\begin{example}[VAR(1) model] \label{example:var1}
Formally, consider an infinite state space $\mathcal{X}=\mathbb{R}^p$. The VAR(1) model with $p$ variates is written as 
    $$x_t = A x_{t-1} + \epsilon_t,\ A \in \mathbb{R}^{p\times p}, \epsilon_t \sim \mathcal{N}(\bm{0},\Sigma)$$
    where $\Sigma$ is the positive-semidefinite covariance matrix. 
    The Markov transition kernel and the stationary distribution are written as 
    \begin{align*}
      w(x_t|x_{t-1}) &= \{(2\pi)^p|\Sigma|\}^{-1/2} \exp{\left(-\frac{1}{2}(x_t-Ax_{t-1})^\top \Sigma^{-1}(x_t-Ax_{t-1})\right)}\\
      &= \exp{(x_t^\top \Sigma^{-1}Ax_{t-1} + \kappa(x_t) - \kappa(x_{t-1}) - \delta(x_t))}\\
      p_w(x) &= (2\pi|B|)^{-1/2} \exp{\left(-\frac{1}{2}x^\top B^{-1} x\right)}
    \end{align*}
    where $\kappa(x) = \frac{1}{2}x^\top A^\top \Sigma^{-1} A x, \delta(x) = x^\top A^\top \Sigma^{-1} A x/2 + x^\top \Sigma^{-1}x/2 + (1/2)\log{\{(2\pi)^p |\Sigma|\}}$.
    The notation $\mathrm{vec}(\cdot)$ denotes the vectorization operator, which transforms a matrix $M$ into a column vector $(M_{11},M_{21},\dots,M_{n-1,n},M_{n,n})^\top$.  Based on the definition of the minimum information Markov model, the corresponding dependence function $\bm{h}: \mathcal{X}^2 \to \mathbb{R}^{p^2}$ is 
    $$\bm{h}(x_t,x_{t-1}) = x_{t} \otimes x_{t-1},$$
and the model parameter becomes $\Theta = A^\top \Sigma^{-1}$.
    It is well-known that the VAR(1) model has a stationary distribution if and only if the eigenvalues of $A$ are smaller than 1 in modulus. 
    This constraint, widely known as the stationarity condition, in fact naturally arises in the parametrization based on the minimum information Markov model. Specifically, we show that these two parametrizations have a one-to-one correspondence via the algebraic Riccati equation, a well-studied equation in control theory.

\begin{definition}[Algebraic Ricatti Equation~\cite{lancaster1995algebraic}]
For a $p\times p$ matrix $F$, 
\begin{equation} \label{eq:algebraic-ricatti}
F^{\top }X+XF-XRX+Q=0
\end{equation}
is called the algebraic Riccati equation, where $Q$ and $R$ are positive semidefinite.
\end{definition}
\begin{lemma}[e.g., \cite{mao2008existence,massei2024data}]
 Eq.~\eqref{eq:algebraic-ricatti} has a unique positive semidefinite solution if $F = -\frac{1}{2}I$ satisfies the \textit{stabilizability condition}: 
    $$\{x^\top F x\ :\ \|x\|_2 = 1, x \in \mathbb{R}^p\} \subset \{z \in \mathbb{C} : \mathrm{Re}(z) < 0\}.$$
\end{lemma}

\begin{proposition}
Let $\mathcal{S}$ denote the set of all $p\times p$ matrices whose eigenvalues are in modulus strictly less than 1.
Let $\mathrm{S}^{p}_{++}$ denote the set of all $p\times p$ positive definite matrices.
    The parameters of the minimum information Markov model $(\Theta, B) \in \mathbb{R}^{p\times p} \times \mathrm{S}^{p}_{++}$ and the stationary VAR(1) parameters $(A, \Sigma) \in \mathcal{S} \times \mathrm{S}^{p}_{++}$ have one-to-one correspondence.
\end{proposition}
\begin{proof}
    The stationary condition of the VAR(1) and the definition of $\Theta$ leads to the following two equations:
    \begin{align} 
    B &= ABA^\top + \Sigma \label{eq:var1-1}\\
    \Theta &= A^\top \Sigma^{-1} \label{eq:var1-2}
    \end{align}

    For given $(A,\Sigma) \in \mathcal{S} \times \mathrm{S}^{p}_{++}$, it is easy to see that $\Theta \in \mathbb{R}^{p\times p}$ and $\mathrm{vec}(B) = (I-A \otimes A)^{-1}\mathrm{vec}(\Sigma)$ exist uniquely. The positive definiteness of $B$ follows from the recursive calculation based on $\rho(A) < 1$ and $\Sigma \in \mathrm{S}^{p}_{++}$:
    $$B = \Sigma + ABA^\top = \Sigma + A(\Sigma + AB A^\top)A^\top = \dots = \sum_{k=0}^\infty A^k \Sigma(A^\top)^k.$$

    Inversely, we show the uniqueness of the solution $(A,\Sigma)$ for the given $(\Theta, B) \in \mathbb{R}^{p\times p} \times \mathrm{S}^{p}_{++}$.  
    Plugging Eq.~\eqref{eq:var1-2} into Eq.~\eqref{eq:var1-1}, the equation becomes
    \begin{equation} \label{eq:var1-3}
    B = \Sigma \Theta^\top B \Theta \Sigma + \Sigma,
    \end{equation}
    which is a special case of the algebraic Ricatti equation with respect to $\Sigma$ (take $F = -\frac{1}{2}I, R = \Theta^\top B \Theta, Q = B$).
\end{proof}

\end{example}

\begin{example}[VAR($d$) model] Again, consider an infinite state space $\mathcal{X} = \mathbb{R}^p$. The VAR$(d)$ model is the extension of the VAR$(1)$ model to the $d$-th order, which is written as 
    $$x_t = \sum_{k=1}^d A_k x_{t-k} + \epsilon_t,\ A_k \in \mathbb{R}^{p\times p}, \epsilon_t \sim \mathcal{N}(\bm{0},\Sigma), $$
    where $\epsilon_t$ is the $p$-dimensional noise term.
    The stationary condition of VAR($d$) model is if all roots of the characteristic equation lie outside the unit circle.
    The Markov kernel and the stationary distribution are given by
    \begin{align*}
      w(x_t|x_{t-1:t-d}) &= \{(2\pi)^p|\Sigma|\}^{-1/2} \exp{\left(-\frac{1}{2}(x_t-\sum_{k=1}^d A_k x_{t-k})^\top \Sigma^{-1}(x_t-\sum_{k=1}^d A_k x_{t-k})\right)}\\
      &= \exp{(\mathrm{vec}(\Theta)^\top (x_t \otimes x_{t-k}) +\kappa(x_{t:t-d+1})-\kappa(x_{t-1:t-d}) - \delta(x_{t}))}\\
    p_w(x) &= (2\pi|B|)^{-1/2} \exp{\left(-\frac{1}{2}x^\top B^{-1} x\right)}
    \end{align*}
    where $\kappa$ and $\delta$ are certain functions. 
    The corresponding dependence function $\bm{h}: \mathcal{X}^{d+1} \to \mathbb{R}^{p^2d}$ is
    $$\bm{h}(x_{t:t-d}) = 
    \begin{pmatrix}
        x_t \otimes x_{t-1}\\
        x_t \otimes x_{t-2}\\
        \vdots \\
        x_t \otimes x_{t-d}\\
    \end{pmatrix}.
    $$
\end{example}

As seen in the part of aforementioned examples, the parameters between the Gaussian AR/VAR models and the corresponding minimum information models are one-to-one.\footnote{Notably, the parameters in the minimum information Markov model are unconstrained. In contrast, the parameters in classical AR/VAR parametrizations are inevitably constrained by the stationarity conditions, and their structure becomes highly complicated when $d\geq 2$.}
These facts suggest the unique existence of the minimum information Markov kernel discussed in Section 2.1. 
Indeed, as illustrated in Examples~\ref{example:ar1}–\ref{example:var1}, the AR(1), AR(2), and VAR(1) cases can each be proven directly. 
Alternatively, this property is guaranteed via the divergence rate minimization problem in Eq.~\eqref{eq:divergence-rate-minimization}, pointing toward a more general framework. For instance, the following statement holds for one-dimensional first order minimum information Markov model. Further extensions to $d \geq 2$ and $p \geq 2$ are discussed in \ref{appendix:var1-and-var2}.

\begin{proposition} \label{prop:ar1}
    Suppose $d=1$ and $p=1$. 
    For the dependence function $h(x,y) = xy$ and the stationary distribution $r(x) = \mathcal{N}(0,\tau^2)$ where $ \tau>0$, the minimum information Markov model \eqref{eq:mininfo-markov-model} exists uniquely. Moreover, it is identical to the Gaussian AR(1) model.
\end{proposition}
\begin{proof}
    See \ref{appendix:ar1}.
\end{proof}

\subsection{Other examples}
Excluding the Gaussian cases, the uniqueness of the minimum information Markov model for an arbitrary dependence function and stationary distribution remains unclear. Formally, however, several existing Markov models can be viewed as instances of the minimum information Markov model, thereby providing a new interpretation of each.

\begin{example}[First-order Gaussian copula Markov model]
The Markov kernel of the first-order Gaussian copula Markov model is defined as 
$$p(x_t|x_{t-1}) = c_{\rho}(G_{t-1}(x_{t-1}), G_t(x_{t}))g_{t}(x_{t})$$
where $c_{\rho}$ is the Gaussian copula density, $G$ and $g$ are the marginal distribution functions and density functions, respectively. 
The corresponding dependence function, which is time-varying here, is 
$$h_t(x_t,x_{t-1}) = \Phi^{-1}(G_t(x_t))\Phi^{-1}(G_{t-1}(x_{t-1}))$$
where $\Phi$ denotes the standard normal distribution function. 
When $G_t = \Phi$ for all $t$, then the model reduces to the AR(1) model $h(x_t, x_{t-1}) = x_tx_{t-1}$. 
In this case, the model parameter becomes $\theta = \frac{\rho}{1-\rho^2}\ (\rho \in [-1,1])$.
\end{example}

\begin{example}[Circular processes]
    A circular AR($d$) process (or CAR($d$) process) is defined directly via the von-Mises distribution as 
    $$X_t | X_{t-1}, X_{t-2}, \dots X_{t-d} \sim \mathrm{vM}(\mu_t, \kappa),$$
    where $\mu_t$ is the mean direction and $\kappa$ is the concentration parameter.  In this case, the state space is a circle $\mathcal{X} = S^1$. The Markov kernel is given by
    $$w(x_t | x_{t-1:t-d}) = (2\pi I_0(\kappa))^{-1}\exp{(\kappa\cos{(x_t-\mu_t)})},$$
    $$\mu_t = \mu + g[\sum_{j=1}^p \alpha_j g^{-1}(x_{t-p} - \mu)],$$
    where $g^{-1}$ is an invertible link function that transforms a circular variable on a real line; $g(\cdot) = 2\mathrm{tan}^{-1}(\cdot)$ for example. This process cannot be expressed within the framework of minimum information Markov model.

    Contrarily, $x$ is defined to be a linked AR process (or LAR) if $g^{-1}(x)$ is an AR process. The dependence function of a linked AR(1) process is simply written as 
    $$h(x_t, x_{t-1}) = g^{-1}(x_t)g^{-1}(x_{t-1}).$$
    
    \noindent These two circular models were proposed in parallel in an early influential paper by Fisher and Lee~\cite{fisher1994time}. 
    The distinction between them becomes clearer in the context of the minimum information Markov model. In each dependence function, the former treats the current value and the previous value asymmetrically, while the latter treats them equally.
    
\end{example}

\begin{example}[Binary process]
    For negative dependence structure lying behind binary sequences, Kanter~\cite{kanter1975autoregression} introduced the binary AR model. 
    A binary AR($d$) process is defined to be the two-state Markov chain $\{X_t\}$ on $\{0,1\}$. The Markov kernel $w:\{0,1\}^2 \to \mathbb{R}_{\geq 0}$ is given by
    $$w(X_t|X_{t-1:t-d}) = \begin{cases}
        l^{-1}(\mu + \sum_{i=1}^d \phi_i X_{t-i}) & (X_t = 1)\\
        1-l^{-1}(\mu + \sum_{i=1}^d \phi_i X_{t-i}) & (X_t = 0)\\
    \end{cases}$$
    where $l$ denotes a link function. Two important cases are the identity link function $l(u) = u$ and the logistic link function $l(u) = \log{(\frac{u}{1-u})}$~\cite{cox2018analysis}. 
    Suppose $d=1$ in the latter case. Then, the model is represented in the form of a minimum information Markov kernel as 
    $$w(x_t|x_{t-1}) = \exp{(\phi_1 x_t x_{t-1} + \kappa(x_t) - \kappa(x_{t-1}) - \delta(x_{t}))}$$
    where $\kappa(x) = \log{(1+e^{\mu + \phi_1 x})}$, $\delta(x) = \kappa(x) - \mu x$.
    In addition, its multivariate extension, which is called binary VAR model, is considered in \cite{agaskar2013alarm} and \cite{shin2019autologistic}, for example.
\end{example}

\begin{example}[Integer-valued AR process]
    See Example 2 of Sei~\cite{sei2024constructingmarkovchainsgiven} for some finite-state Markov chain examples.
\end{example}

\section{Parameter estimation}
In this section, we explore the methods to estimate the model parameter $\bm{\theta}$ in Equation~\ref{eq:joint}.
We suggest an estimation based on conditional likelihood, leveraging the fact that the likelihood function derived from Equation~\ref{eq:joint} is invariant to a certain type of permutation on data.
The advantage of using conditional likelihood is that it does not require the information about marginal distributions $\kappa$ and $\delta$, which are usually intractable analytically. On the other hand, Besag's likelihood~\cite{besag1975statistical} is also considered for comparison.

\subsection{Maximum conditional likelihood estimation}

For simplicity, we consider stationary $d$-th-order univariate minimum information Markov model to demonstrate the maximum conditional likelihood estimation, although the extension to multivariate case is straightforward. 
Suppose we have the observation $(x_1, \dots, x_n)$.
It turns out that the likelihood function 
$$L(\bm{\theta}) = L(\bm{\theta},x) = \prod_{i=1}^n p(x_t|x_{t-1:d}) \simeq \exp{\left(\sum_{t=d+1}^n \bm{\theta}^\top \bm{h}(x_{t:t-d}) + \kappa(x_{n:n-d+1}) - \kappa(x_{d:1}) - \sum_{t=d+1}^n \delta(x_t) \right)}$$
allows us to ignore the unknown normalizing functions $\kappa$ and $\delta$ when we consider the following permutation. 
Consider a permutation of $(x_1, \dots, x_n)$ where the first $d$ elements $(x_1,\dots,x_d)$ and the last $d$ elements $(x_{n-d+1},\dots, x_n)$ are fixed but the rest follow an $(n-2d)$-th-order permutation. Let $\Pi$ denote such a set. 
Then, based on its definition, the terms other than the dependence term $\bm{\theta}^\top \bm{h}(x_{t:t-d})$ remain unchanged under this permutation and can therefore be ignored by leveraging conditional inference on it. 
Note that $\Pi$ is different from the symmetric group of degree $n$.
Now, we formulate the estimator by conditional inference.

\begin{lemma}
    The likelihood function is decomposed as
    $$L(\bm{\theta},x) = f(\pi=\mathrm{id}|\bm{\theta}) \sum_{\tilde{\pi}\in \Pi} L(\bm{\theta},\tilde{\pi}\circ x),$$
    \begin{equation}\label{eq:cl}
    f(\pi|\bm{\theta}) = \frac{\tilde{L}(\bm{\theta}, \pi\circ x)}{\sum_{\tilde{\pi} \in \Pi} \tilde{L}(\bm{\theta}, \tilde{\pi} \circ x)}
    \end{equation}
 where $\tilde{L}(\bm{\theta},x) = \exp{(\sum_{t=d+1}^n \bm{\theta}^\top \bm{h}(x_{t:t-d})})$.
    Here, $f(\pi|\bm{\theta})$ is the conditional likelihood given a set $\{x_1, \dots, x_n\}$ without its ordering.
\end{lemma}

\begin{definition}
Given the data $(x_1,\dots,x_n)$, 
the conditional maximum likelihood estimator (CLE) $\hat{\bm{\theta}}_{CLE}$ is defined as a maximizer of the conditional likelihood \eqref{eq:cl}:
\begin{equation}\label{eq:cle}
\hat{\bm{\theta}}_{CLE} = \mathrm{argmax}_\theta \ 
 f(\pi=\mathrm{id}|\bm{\theta}).
\end{equation}

\end{definition}
\noindent The estimator $\hat{\theta}_{CLE}$ is calculated in line with Sei and Yano~\cite{sei2024minimum} using Algorithm~\ref{alg:exchange}, which samples from $f(\pi|\bm{\theta})$ via Metropolis-Hasting sampling and applies the Fisher scoring method.
In practice, the swap between two entries $\tau_{st}$ is convenient instead of $\Pi$ in terms of stability, which does not change the likelihood too much from $f(\pi)$ to $f(\pi \circ \tau_{st})$. In this case, the acceptance ratio $\rho$ can be easily calculated by only considering the swapped data. 
In each iteration in Fisher scoring (Algorithm~\ref{alg:fisher-scoring}) when updating $\bm{\theta}$, it is required to implement the Fisher information matrix fror $f(\pi|\bm{\theta})$:
$$G_{jk}(\bm{\theta}) = \mathbb{E}[\{\sum_{t=d+1}^n \bm{h}_j(({\pi} \circ x)_{t:t-d}) - \mu_j(\bm{\theta})\}\{\sum_{t=d+1}^n \bm{h}_k(({\pi} \circ x)_{t:t-d}) - \mu_k(\bm{\theta})\}],$$
$$\mu_i(\bm{\theta}) = \mathbb{E}[\sum_{t=d+1}^n \bm{h}_i(({\pi} \circ x)_{t:t-d})]$$
where the expectation is taken with respect to $f(\pi|\bm{\theta})$. The empirical approximation of it can be calculated by using the MCMC samples from Algorithm~\ref{alg:exchange}. 

\begin{figure}[!t]
\begin{algorithm}[H]
    \caption{Exchange algorithm (Metropolis-Hasting sampling)}
    \label{alg:exchange}
    \begin{algorithmic}[1] 
    \REQUIRE An initial permutation $\pi(0)$ and the number of samples $L$.\\
    \STATE Step 1: Initialize $\pi \leftarrow \pi(0)$ and $l = 1$.\\
    \STATE Step 2: Select $1 \leq s < t \leq n-1$ uniformly at random. Let $\tau_{st}$ be the transposition between $x_s$ and $x_t$. Compute the conditional likelihood ratio $\rho = \frac{f(\pi=\tau_{st}|\bm{\theta})}{f(\pi=\mathrm{id}|\bm{\theta})}$.\\
    \STATE Step 3: Generate a random number $u$ uniformly distributed on $[0, 1]$. 
    If $u \leq \min(1, \rho)$, then update $\pi$ to $\pi\circ\tau_{st}$.\\
    \STATE Step 4: Let $\pi(l) \leftarrow \pi$ and $l \leftarrow l + 1$. Go to Step 2 if $l \leq L$, and output $(\pi(l))_{l=1}^L$ otherwise.
    \end{algorithmic}
\end{algorithm}

\begin{algorithm}[H]
    \caption{Fisher scoring}
    \label{alg:fisher-scoring}
    \begin{algorithmic}[1] 
    \REQUIRE An initial estimation $\theta_0$.\\
    \STATE Step 1: Initialize $\hat{\theta} \leftarrow \theta_0$.\\
    \WHILE{not converge}
    \STATE $S \leftarrow$ Exchange algorithm($\hat{\theta})$ $\#\ S:L$ samples of permutations
    \STATE Discard the beginning of $S$. \# Burn-in 
    \STATE $H[l] \leftarrow \sum_{t=d+1}^n \bm{h}((S[l]\circ x)_{t:t-d})$ 
    \STATE $\hat{\mu} \leftarrow \sum_{l=1}^L H[l] / |S|$ \ \\ 
    \STATE $\hat{G} \leftarrow (\sum_{jk} (H_j -\hat{\mu}_j)(H_k -\hat{\mu}_k)^\top)/|S|$ \ \# Fisher information matrix
    \STATE $\hat{\theta} \leftarrow \hat{\theta} + \hat{G}^{-1} \hat{\mu}$
    \ENDWHILE
    \end{algorithmic}
\end{algorithm}
\end{figure}

\subsection{Maximum Besag's pseudo likelihood estimation}

As a baseline for comparison, we consider the maximum Besag's pseudo likelihood estimation~\cite{besag1975statistical} (abbreviated as PLE hereafter). When the genuine likelihood function is complicated, Besag's likelihood (or pseudo likelihood) is well-known as an approximation of it which provides a computationally simpler statistical estimation.

Consider $n$ observations $\{x(1),\dots,x(n)\}$. Approximating Equation~\eqref{eq:cle}, the Besag's likelihood is written as 

\begin{align*}
L_{PLE}(\bm{\theta}) = L_{PLE}(\bm{\theta},x) &= \prod_{d<s_1<s_2<n-d} \frac{\tilde{L}(\bm{\theta},\pi\circ x)}{\tilde{L}(\bm{\theta},\pi\circ x) + \tilde{L}(\bm{\theta},\pi\circ\tau_{s_1,s_2}\circ x)}\\
&= \prod_{d<s_1<s_2<n-d} \frac{1}{1 + \exp{(-\bm{\theta}^\top \sum_{t=d+1}^n  \{\bm{h}((\pi\circ x)_{t:t-d})- \bm{h}((\pi\circ\tau_{s_1,s_2}\circ x)_{t:t-d})}\})}
\end{align*}

Note that this likelihood is the same form as the logistic regression without interception and without penalty where the explanatory variable is 
\begin{equation} \label{eq:PLE_X}
X_{(s_1,s_2)} = \sum_{t=d+1}^n  \{\bm{h}((\pi\circ x)_{t:t-d})- \bm{h}((\pi\circ\tau_{s_1,s_2}\circ x)_{t:t-d})\}
\end{equation}
and the response variable is always 1. Therefore, the estimator $\hat{\bm{\theta}}_{PLE} = \mathrm{argmax}_{\bm{\theta}}\ L_{PLE}(\bm{\theta})$ can be obtained by implementing the logistic regressor.

For an optimization method to maximize the Besag's pseudo likelihood $L_{PLE}$, the standard gradient descent method is utilized as a baseline, which we refer to as the naive PLE hereafter. However, the optimization becomes more difficult when the dimension of data increases and more time consuming under large samples. Also, the memory issue arises because storing $X \in \mathbb{R}^{_{(n-2d)}C_2 \times d}$ is required. To alleviate these issues, we apply two subsampling strategies to make it work under practical settings. The comparison between the performance of these two optimization methods will be presented in the following Section 4 using simulation data generated from a variety of pre-specified autoregressive models.

\section{Simulation studies with Gaussian autoregressive processes}

In this section, we examine the workability of the estimation methods for our minimum information Markov model using artificially generated data from various Gaussian autoregressive models.

\subsection{MLE vs MCLE vs PLE}
First, as the simplest example, we consider several stationary Gaussian AR and VAR processes. We compare the proposed estimators, MCLE and PLE, in Table~\ref{tab:performance}. 
In addition, in these cases it is well-known that the maximum likelihood estimators (MLEs), with the initial $d$ observations $x_{d:1}$ treated as fixed and given, are available in closed form via the ordinary least squares. Therefore, we also report the MLE results for reference. However, note that MLE is not always tractable for other models. 

Table~\ref{tab:performance} shows the mean Euclidean norm of the estimation error over 30 repetitions and the estimation time per run in seconds, shown as (error, time) in each cell.
Notably, it is observed that the computation burden of MCLE is high and increases drastically as the sample size or the model dimension (both $d$ and $p$ in Definition 2) increases. Algorithm~\ref{alg:exchange} and \ref{alg:fisher-scoring} suggest that this burden is attributed to two reasons. Firstly, the sampling via exchange algorithm becomes inefficient, especially in high dimensional settings. Figure~\ref{fig:acc_rate} shows the sharp decline in acceptance ratio of the Metropolis Hasting sampling in Algorithm~\ref{alg:exchange}. 
\begin{figure}
    \centering
    \includegraphics[width=0.5\linewidth]{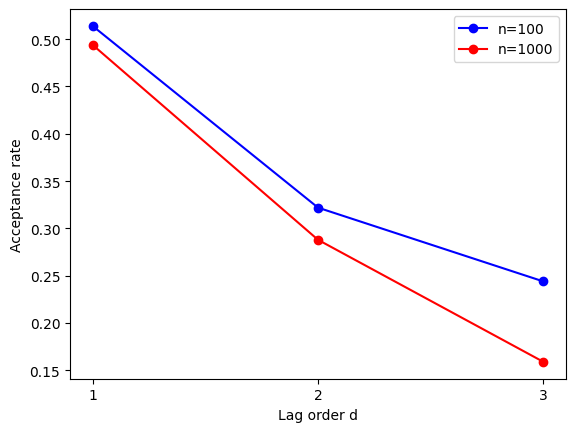}
    \caption{Decline in the acceptance rate of exchange algorithm during MCLE as the order of AR model increases.}
    \label{fig:acc_rate}
\end{figure}
Contrarily, the estimation errors of PLE were comparable to MLE and significantly better than MCLE.
The optimization procedure for the logistic regressor in PLE initialized the weight parameters to zero and employed a learning rate with an inverse time decay schedule. Despite the overall computational time, the optimization accounted for only about 6\%, with the primary bottleneck being the computation of the explanatory variable $X$ in Eq.~\eqref{eq:PLE_X}.

\begin{table}
    \centering
    \begin{subtable}{\linewidth}
    \scalebox{1}{
    \begin{tabular}{|cc|c|ccc|} \hline
    True model & True parameter & Sample size & MLE(ref.) & MCLE & PLE\\\hline
    AR(1) & $\phi=0.5, \sigma^2=0.5, \theta^*=1$ & 100 & 0.208 & 0.310 & 0.211\\
    AR(1) & $\phi=0.5, \sigma^2=0.5, \theta^*=1$ & 1000 & 0.0625 & 0.0774 & 0.0692\\ \hdashline
    AR(2) & $\bm{\phi}=(0.5,0.3)^\top, \sigma^2=0.5, \bm{\theta}^*=(0.7,0.6)^\top$ & 100 & 0.294 & 0.341 & 0.199\\
    AR(2) & $\bm{\phi}=(0.5,0.3)^\top, \sigma^2=0.5, \bm{\theta}^*=(0.7,0.6)^\top$ & 1000 & 0.0869 & -- & 0.0876\\ \hdashline
    AR(3) & $\bm{\phi}=(0.5,0.3,0.1)^\top, \sigma^2=0.5, \bm{\theta}^* = (0.64,0.5,0.2)^\top$ & 100 & 0.315 & 0.472 & 0.279\\
    AR(3) & $\bm{\phi}=(0.5,0.3,0.1)^\top, \sigma^2=0.5, \bm{\theta}^* = (0.64,0.5,0.2)^\top$ & 1000 & 0.120 & -- & 0.196\\  \hdashline
    VAR(1)& $A = \begin{pmatrix}
        0.5&0.1\\
        0.1&0.5
    \end{pmatrix}, \Sigma=0.5I, \Theta^* = \begin{pmatrix}
        1&0.2\\
        0.2&1
    \end{pmatrix}$ & 100 & 0.372 & 0.607 & 0.547\\
    VAR(1)& $A = \begin{pmatrix}
        0.5&0.1\\
        0.1&0.5
    \end{pmatrix}, \Sigma=0.5I, \Theta^* = \begin{pmatrix}
        1&0.2\\
        0.2&1
    \end{pmatrix}$ & 1000 & 0.124 & -- & 0.132\\ \hline
    \end{tabular}}
    \subcaption{Mean Euclidean norm of the estimation error $\|\bm{\hat{\theta}}-\bm{\theta}^*\|_2$ over 30 repetitions ($\bm{\theta}=\mathrm{vec}(\Theta))$. }
    \end{subtable}

    \vspace{1em}
    
    \begin{subtable}{\linewidth}
    \scalebox{1}{
    \begin{tabular}{|cc|c|ccc|} \hline
    True model & True parameter & Sample size & MLE(ref.) & MCLE & PLE\\\hline
    AR(1) & $\phi=0.5, \sigma^2=0.5, \theta^*=1$ & 100 & 0.0145 & 79.32 & 0.320\\
    AR(1) & $\phi=0.5, \sigma^2=0.5, \theta^*=1$ & 1000 & 0.0311 & 830.14 & 260.1\\ \hdashline
    AR(2) & $\bm{\phi}=(0.5,0.3)^\top, \sigma^2=0.5, \bm{\theta}^*=(0.7,0.6)^\top$ & 100 & 0.0954 & 209.29 & 0.475\\
    AR(2) & $\bm{\phi}=(0.5,0.3)^\top, \sigma^2=0.5,\bm{\theta}^*=(0.7,0.6)^\top$ & 1000 & 0.0493 & -- & 392.50\\ \hdashline
    AR(3) & $\bm{\phi}=(0.5,0.3,0.1)^\top, \sigma^2=0.5, \bm{\theta}^* = (0.64,0.5,0.2)^\top$ & 100 & 0.0289 & 353.98 & 12.881\\
    AR(3) & $\bm{\phi}=(0.5,0.3,0.1)^\top, \sigma^2=0.5, \bm{\theta}^* = (0.64,0.5,0.2)^\top$ & 1000 & 0.0743 & -- & 373.49\\ \hdashline
    VAR(1) & $A = \begin{pmatrix}
        0.5 & 0.1\\
        0.1 & 0.5
    \end{pmatrix}, \Sigma=0.5I, \Theta^* = \begin{pmatrix}
        1 & 0.2\\
        0.2 & 1
    \end{pmatrix}$ & 100 & 0.00583 & 208.645 & 0.540\\
    VAR(1) & $A = \begin{pmatrix}
        0.5 & 0.1\\
        0.1 & 0.5
    \end{pmatrix}, \Sigma=0.5I, \Theta^* = \begin{pmatrix}
        1 & 0.2\\
        0.2 & 1
    \end{pmatrix}$ & 1000 & 0.00642 & -- & 519.13\\ \hline
    \end{tabular}}
    \subcaption{Estimation time per run in seconds (bottom). The record ``--'' indicates that the computation per run did not terminate within the time limit of 15 minutes.}
    \end{subtable}
    
    \caption{Comparison in performance of MLE, MCLE, and PLE with synthetic data generated from pre-specified AR/VAR models. In PLE, the logistic regression optimization procedure initialized the weight parameters to zero and used a learning rate with inverse time decay schedule.  }
    \label{tab:performance}
\end{table}

\subsection{Improving PLE}

Regarding the fact that PLE cannot even work well for simple AR process, PLE seems to be no good choice. However, we still argue that PLE is a promising choice in practice. While the estimation performance of PLE is worse than MCLE, the estimation time per run is more preferable. For example in VAR(1) using 1000 samples, MCLE took more than 16 hours just for a single run, whereas PLE achieved the estimation within 9 minutes. It is clear from these experimental results that MCLE is not feasible but PLE is quite acceptable in practice, balancing the tradeoff between the estimation error and time. 
In real scenario, the length of time series data could be significantly larger than 1000. Even PLE could be infeasible by naive calculation due to the exponential growth in computation time with respect to the sample size, as in Figure~\ref{fig:ple_time}. 
Here we propose and examine additional methods that further accelerate PMLE.

\begin{figure}[t]
\centering
    \centering
    \includegraphics[width=0.49\columnwidth]{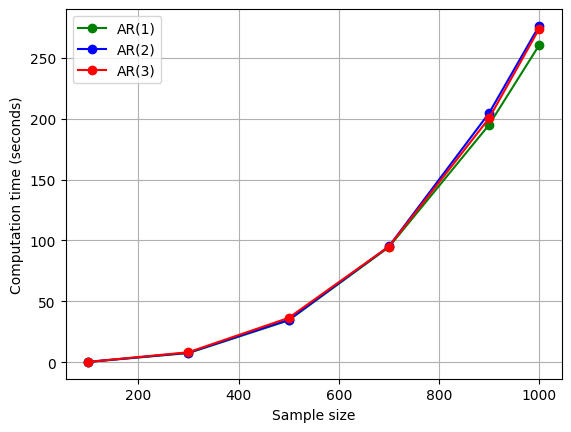}
    \caption{Estimation time of PLE with synthetic data generated from each AR process across different sample sizes.}
    \label{fig:ple_time}
\end{figure}

\paragraph{Bipartition approach}
Firstly, inspired by the conditional KL score minimization in Chen and Sei~\cite{chen2024proper}, we randomly split the observed data into two groups and make a pair of data points. Then, the corresponding terms in $L_{PLE}(\theta)$ are only calculated. In this way, the number of terms is reduced to $\lfloor (n-2d)/2\rfloor \sim O(n)$ from $_{(n-2d)}C_2 \sim O(n^2)$.

\paragraph{Online SGD approach}
Secondly, inspired by prior studies on large sample logistic regression~\cite{wang2018optimal}, we introduce subsampling-based approaches. Specifically, we implement an online stochastic gradient descent method that prevents us from calculating a large data matrix $X_{st} = (h(\pi)-h(\pi\circ\tau_{s,t})) \in \mathbb{R}^{_{(n-2d)}C_2 \times d}$, which costs us long computation time and large memory. Its pseudo code is exhibited in Algorithm~\ref{alg:online-SGD}. 
In the simulation using AR(1) model, Table~\ref{tab:pmle_online-SGD} shows that the online SGD approach performs better with a small learning rate $\eta$ and a large iteration number, which almost matches the performance of the naive PLE presented in Table~\ref{tab:performance}.\\

To compare these approaches in practical settings, we conducted a series of experiments similar to the previous one with large sample size. In the online SGD approach, the learning rate $\eta$ was fixed to 0.01 and the number of iteration was fixed to 10000 to balance the estimation performance and the computational time according to the results shown in Table~\ref{tab:pmle_online-SGD}. Table~\ref{tab:pmle-comparison} summarizes the performance of each method.
Compared to the naive PLE, both the bipartition approach and the online SGD approach show significant acceleration. Specifically, the former outperforms the latter at sample sizes of 1000 and 10000.
The online SGD method demonstrates better acceleration when the sample size reaches 100000, while the estimation error remains empirically comparable.

\begin{algorithm}[t]
    \caption{Online SGD for large sample logistic regression}
    \label{alg:online-SGD}
    \begin{algorithmic}[1] 
    \REQUIRE An initial estimation $\theta_0 \in \mathbb{R}^K$. Learning rate $\eta \in \mathbb{R}_+$.\\
    \WHILE{not converge}
    \STATE Select $d < s_1 < s_2 \leq n-d$ uniformly at random. 
    \STATE $x = \sum \{h(\pi) - h(\pi\circ\tau_{s_1,s_2})\}$
    \STATE $\theta \leftarrow \theta + \eta (1-1/\{1+\exp{(-\theta^\top x)}\})$
    \ENDWHILE
    \end{algorithmic}
\end{algorithm}

\begin{table}[t]
    \centering
    \begin{tabular}{c|cccc}\hline
         &  \#iter=$10^3$& \#iter=$10^4$ & \#iter=$10^5$&\#iter=$10^6$\\ \hline
         $\eta = 0.001$& (0.649, 0.580) & (0.127, 5.662) & (0.0752, 58.447)& (0.0734,576.87)\\
         $\eta = 0.01$& (0.112,0.564) & (0.0892,5.651) & (0.0866,59.142)& (0.0769, 577.18)\\
         $\eta = 0.1$&  (0.229,0.565) & (0.236,6.029) & (0.211,56.63)&(0.189,592.07)\\ \hline
    \end{tabular}
    \caption{The performance of online SGD approach on synthetic data from the AR(1) model. Values in each cell represent the mean Euclidean norm of the estimation error $\|\bm{\hat{\theta}}-\bm{\theta}^*\|_2$ over 30 repetitions and the estimation time per run in seconds. Performance is better with a small learning rate $\eta$ and a large iteration number (\#iter), which almost matches the performance of the naive PLE presented in Table~\ref{tab:performance}. The data used is 1000 samples generated from AR(1) model: $x_{t+1} = 0.5x_t + \epsilon_t,\ \epsilon_t \sim \mathcal{N}(0,0.5)$. }
    \label{tab:pmle_online-SGD}
\end{table}

\begin{table}[t]
    \centering
    \begin{tabular}{|ccc|cc|} \hline
        True model & True parameter & Sample size & Bipartition & Online SGD ($\eta=0.01$, \#iter=$10^4$)\\\hline
        AR(1) & $\phi=0.5, \sigma^2=0.5$ & 1000 & (0.0989,\textbf{0.265})& (\textbf{0.0892},5.651)\\
        AR(1) & $\phi=0.5, \sigma^2=0.5$ & 10000 & (\textbf{0.0327},\textbf{25.587})& (0.0773, 54.214)\\
        AR(1) & $\phi=0.5, \sigma^2=0.5$ & 100000 & (\textbf{0.00901},3059.74)& (0.0620,\textbf{1111.234})\\ \hdashline
        AR(2) & $\bm{\phi}=(0.5,0.3)^\top, \sigma^2=0.5$ & 1000 & (0.311,\textbf{0.291})& (\textbf{0.299},5.693)\\
        AR(2) & $\bm{\phi}=(0.5,0.3)^\top, \sigma^2=0.5$ & 10000 & (\textbf{0.3008},\textbf{25.565})& (0.317,53.963)\\
        AR(2) & $\bm{\phi}=(0.5,0.3)^\top, \sigma^2=0.5$ & 100000 & (\textbf{0.298},3180.14)& (0.306,\textbf{1195.533})\\ \hdashline
        AR(3) & $\bm{\phi}=(0.5,0.3,0.1)^\top, \sigma^2=0.5$ & 1000 & (0.457,\textbf{0.280})&(\textbf{0.382},5.841)\\
        AR(3) & $\bm{\phi}=(0.5,0.3,0.1)^\top, \sigma^2=0.5$ & 10000 & (0.389,\textbf{25.877})&(\textbf{0.381},55.389)\\ 
        AR(3) & $\bm{\phi}=(0.5,0.3,0.1)^\top, \sigma^2=0.5$ & 100000 & (\textbf{0.373},3170.67)&(0.383,\textbf{931.57})\\ \hline 
    \end{tabular}
    \caption{Comparison of the bipartition approach and the online SGD approach for calculating PLE. The better scores in each row are bolded.}
    \label{tab:pmle-comparison}
\end{table}

\section{Real time-series data applications}

Let us review the framework of the minimum information Markov modeling. When you have a stationary time series data, you specify two factors in principle: (i) the stationary distribution and (ii) the dependence function $h(x_{t:t-d})$. The order of Markov kernel $d$ can be taken arbitrarily here. In practice, our proposed method for estimating dependence does not involve the estimation of the marginal (stationary) distribution. Consequently, only the form of the dependence function $h$ needs to be specified. The marginal distributions can, if desired, be estimated separately, for example, using empirical distributions. 
The dependence function can be interpreted as the strength of interaction effect between $x_t$ and $x_{t-d}$ for $d \geq 1$. 
The simplest choice is $d=1$ and $h(x_t,x_{t-1}) = x_t x_{t-1}$, which leads to the same dependence structure with AR(1) model, as seen in Example~\ref{example:ar1}. Differently, you may pick $h(x_t, x_{t-1}) = x_t^a x_{t-1}^b\ (a,b \in \mathbb{N})$ which is similar to step-wise inclusion of basis functions in \cite{bedford2016approximate}, or increase the dimension of output from the dependence function $h$. 
Once $h$ is specified, the coefficients corresponding to each $h$ are estimated by maximizing the conditional likelihood or approximately maximizing the pseudo likelihood, as in Section 3. Due to the unified expression of the minimum information Markov model, which includes various typical models, the estimation procedure can also be conducted in a unified manner regardless of the data or any assumptions. 
This is in contrast with the standard maximum likelihood procedure, where the parametric model is pre-specified and the algorithm for obtaining the maximum likelihood estimator is proposed per model; when the maximum likelihood estimator is intractable, for instance in energy-based models, the alternative estimators like the score matching estimator are studied.

\begin{remark}
For example, the model
$$x = \{x_1,x_2,\dots,x_n\}$$
$$h(x_t,x_{t-1}) = (x_tx_{t-1} , x_t^2 x_{t-1} , x_t x_{t-1}^2)^\top$$
$$\theta = (\theta_1,\theta_2,\theta_3)^\top$$
seems to be equivalent to the restricted VAR(1)-alike model:
$$x' = \{(x_1,x_1^2)^\top, (x_2,x_2^2)^\top,\dots,(x_n,x_n^2)^\top\};$$
$$h(x'_t,x'_{t-1}) = x'_t \otimes x'_{t-1} = (x_tx_{t-1}, x_tx_{t-1}^2, x_t^2x_{t-1}, x_t^2x_{t-1}^2)^\top;$$
$$\theta' = (\theta_1, \theta_3, \theta_2, 0)^\top.$$
However, since the marginal distribution of $x'$ is obviously not the bivariate Gaussian, this model is not exactly the bivariate VAR(1) model, although they share the same dependence function $h$.
\end{remark}

\subsection{Dataset acquisition}

We apply our method on time series analysis on V4 Utah Array Plaid Movie Data~\cite{Smith2020}. 
This dataset contains two data files from individual subjects performing a fixation task. Four types of data were recorded simultaneously on 96 electrode channels.
Specifically, we analyze the local field potentials and the spike trains of one individual, refered to as ``Wi170428''. 
Local field potentials (LFPs) and spike trains capture two complementary aspects of neural activity.
LFPs are low-frequency extracellular signals that reflect the summed synaptic input and subthreshold activity of populations of neurons within a local region. They are continuous-valued signals and typically analyzed in the frequency domain (e.g., using power spectral density), capturing oscillatory dynamics and collective network states.
In contrast, spike trains are point processes that encode the timing of action potentials emitted by individual neurons. They are sparse, high-temporal-resolution representations of neural output, often analyzed in terms of firing rates, inter-spike intervals, or spike-triggered average

\begin{remark}
Since both LFPs and spike trains are important sources of information in neuroscience, many studies have been conducted on each separately.
Some studies have utilized autoregressive models for the analysis of LFPs~\cite{huberdeau2011analysis, hoerzer2010directed}.
In contrast, spike trains are often naturally modeled as inhomogeneous Poisson processes or other point processes, and the use of AR models for spike trains remains limited.
Beyond modeling each signal independently, relating these two types of data, referred to as spike–LFP coupling in some literature, is also of great interest for revealing the mechanisms of neural information processing.
Several studies have proposed models and methods to relate LFP dynamics to spike trains~\cite{gutnisky2010generation, banerjee2012parametric}, including approaches based on state-space models.
In this study, we use such data as an application example for our method; however, we do not aim to compare the performance of different models from perspectives from neuroscience or to claim that our modeling choice is the most appropriate. These aspects are beyond the scope of this work.
\end{remark}

\subsection{A univariate example on local field potential}

Firstly, we apply the minimum information Markov model with four different dependence functions to univariate LFP series, preprocessed by standard scaling, where $N=1000$ (Figure~\ref{fig:LFP}). 
As baselines, we pick the choices that share equivalent dependence structure with the Gaussian AR(1) and AR(2) models:
$$h(x_t,x_{t-1}) = (x_tx_{t-1})$$
and
$$h(x_t,x_{t-1}) = (x_tx_{t-1}, x_tx_{t-2})^\top.$$
By extending the latter, we also employ as the third candidate
$$h(x_t,x_{t-1}) = (x_tx_{t-1}, x_tx_{t-2}, x_tx_{t-1}x_{t-2})^\top.$$
Furthermore, we add
$$h(x_t,x_{t-1}) = (x_tx_{t-1}, x_t^2x_{t-1}, x_tx_{t-1}^2)^\top$$
as the fourth candidate.
The first entry is identical to the AR(1) model, however, the second and third cannot be expressed with it, aiming at capturing further flexible temporal dependence structure. 

\begin{figure}[t]
    \centering
    \includegraphics[width=1\linewidth]{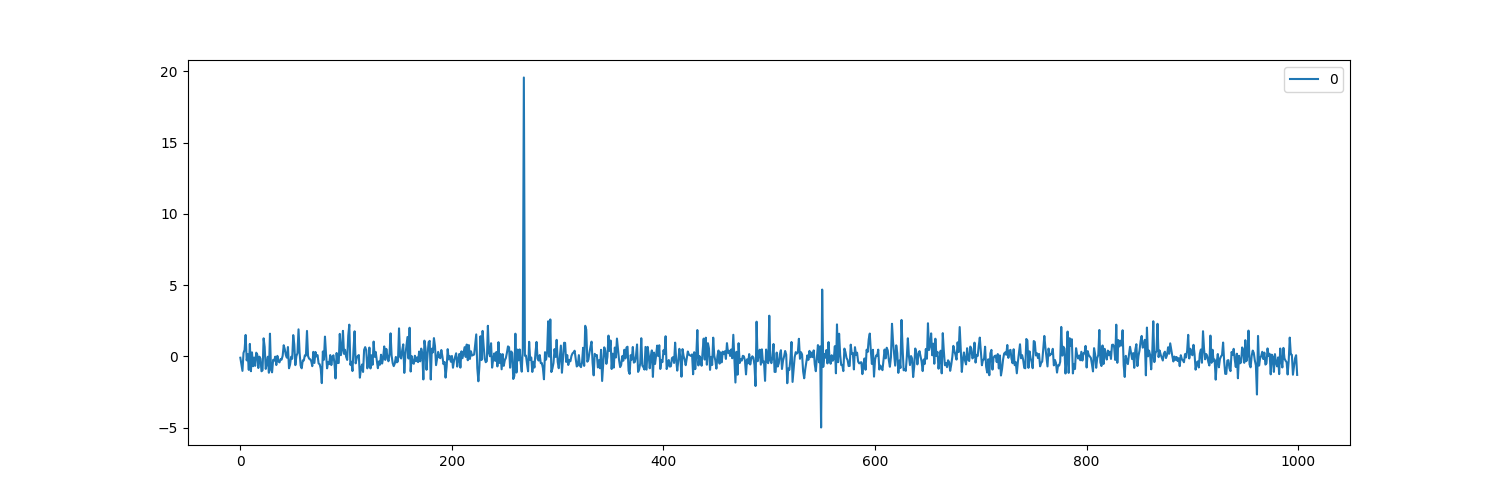}
    \caption{An example of standard scaled LFP; channel = 1.}
    \label{fig:LFP}
\end{figure}


It is also possible to compare different models in line with the usual AIC-based model selection procedure. On the other hand, the pseudolikelihood-based information criterion for Markov random fields was proposed by Ji and Seymour~\cite{ji1996consistent}:
$$\mathrm{AIC} = -2 \log{L_{PLE}}(\bm{\hat{\theta}_{PLE}}) + 2\mathrm{dim}{(\bm{\theta})};$$
$$\mathrm{PIC} = -2 \log{L_{PLE}}(\bm{\hat{\theta}_{PLE}}) + \mathrm{dim}{(\bm{\theta})}\log{_{N-2d}C_2}.$$

Table~\ref{tab:PLE-univariate-LFP} shows the results of PLE performed on a standard-scaled univariate LFP data of length 1000. The comparison on AIC indicates that $h(x_t,x_{t-1}) = x_tx_{t-1}$, which has the same dependence structure with the Gaussian AR(1) model, is the best fit for this data.

\begin{table}
    \centering 
    \begin{tabular}{ccccc} \hline
         $\bm{h}(x_t,x_{t-1})$& $\bm{\hat{\theta}}_{PLE}$ & $\log{L_{PLE}(\bm{\hat{\theta}}_{PLE})}$&AIC&  PIC\\ \hline
         $(x_tx_{t-1})^\top$ & $(-0.0202)^\top$ & -343262.87 & \bf{686527.75}& \bf{686538.85}\\
         $(x_tx_{t-1}, x_t^2x_{t-1}, x_tx_{t-1}^2)^\top$& $(-0.269,  0.802, -0.0643)^\top$ & -520619.75 & 1041245.50 & 1041278.82\\
         $(x_tx_{t-1}, x_tx_{t-2})^\top$ &$(-0.0493,  0.00364)^\top$ & -343679.34 & 687362.68 & 687384.90\\
         $(x_tx_{t-1}, x_tx_{t-2}, x_tx_{t-1}x_{t-2})^\top$& $(-0.0648,0.0119, 0.00316)^\top$&-344425.06 &688856.13  & 688889.44\\ \hline
    \end{tabular}
    \caption{PLE on univariate LFP dataset (channel = 1). The first row is equivalent to Gaussian AR(1). The third row is equivalent to Gaussian AR(2).}
    \label{tab:PLE-univariate-LFP}
\end{table}

\subsection{Cross domain analysis between neuro spike train and local field potential}

Secondly, we demonstrate that our proposed method is applicable to a pair of time series data from different domains.
Figure~\ref{fig:LFP-spike} shows an example of simultaneously recorded LFP and spike trains from V4 Utah dataset. 

\begin{figure}
    \centering
    \includegraphics[width=1\linewidth]{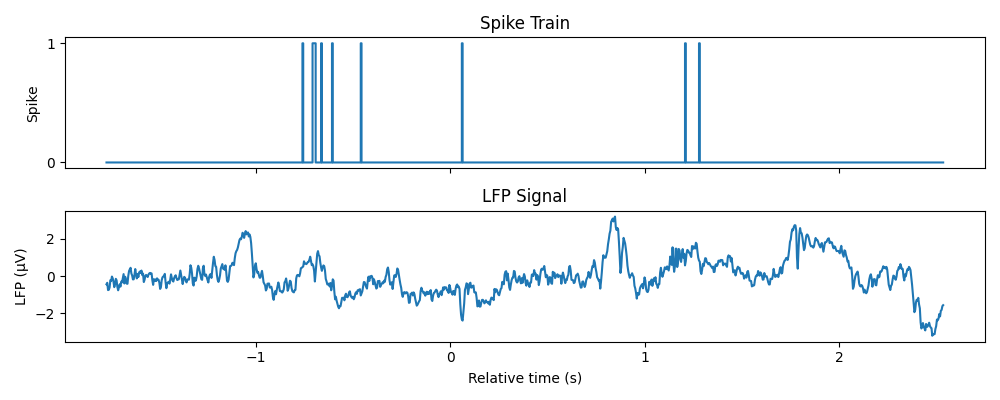}\\
    \includegraphics[width=1\linewidth]{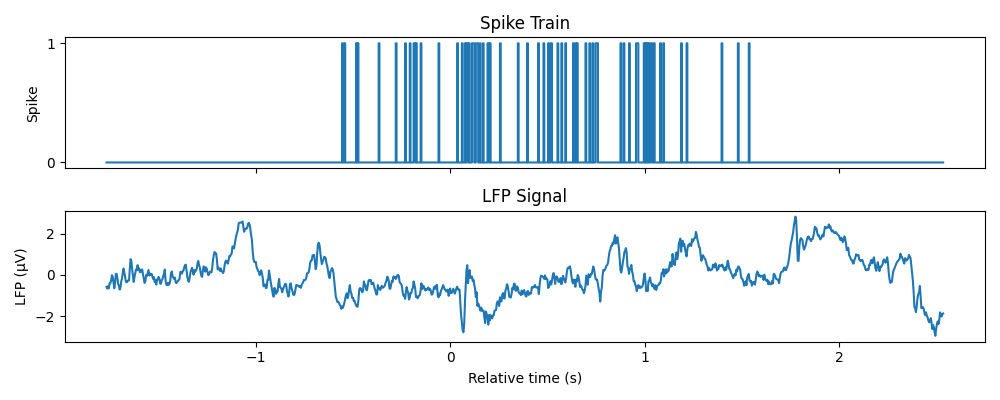}
    \caption{An example of simultaneously recorded LFP after standard scaling and binary spike trains; top: channel = 2, bottom: channel = 75. The length of each data is 1007. The time axis shows the time relative to the stimulation offset at 0.}
    \label{fig:LFP-spike}
\end{figure}

Suppose $x$ is a series of bivariate data points. Without loss of generality, suppose the first entry is the binary variable and the second entry is the real-valued variable, i.e., $x_t = (x_t^{bin}, x_t^{real})^\top$. 
In the usual VAR(1) modeling, the dependence function is modeled as
$$h(x_t,x_{t-1}) = x_t \otimes x_{t-1} = \begin{pmatrix}
    x_{t}^{bin}x_{t-1}^{bin}\\
    x_{t}^{bin}x_{t-1}^{real}\\
    x_{t}^{real}x_{t-1}^{bin}\\
    x_{t}^{real}x_{t-1}^{real}
\end{pmatrix}$$
as seen in Example~\ref{example:var1}.
The first entry indicates the binary AR(1) modeling for $x^{bin}$ and the last entry indicates the Gaussian AR(1) modeling for $x^{real}$. 
To enable capturing more complex dependence structure, we could further extend the dependence function straightforwardly to
$$h(x_t,x_{t-1}) = \begin{pmatrix}
    x_{t}^{bin}x_{t-1}^{bin}\\
    x_{t}^{bin} (x_{t-1}^{real})^{b_1}\\
    (x_{t}^{real})^{a_1}x_{t-1}^{bin}\\
    (x_{t}^{real})^{a_2}(x_{t-1}^{real})^{b_2}
\end{pmatrix},\ a_1, a_2 \in \mathbb{R}, b_1, b_2 \geq 1.$$

Table~\ref{tab:PLE-bivariate-LFP} shows the results of PLE performed on a bivariate LFP-Spike data of length 1007. The comparison of AIC/PIC indicates that the most complex of the four candidates was selected as the best fit for this data, suggesting its potential to capture the complex dependencies underlying the data through the flexible dependence function that the minimum information Markov model can express.

\begin{table}[t]
    \centering 
    \begin{tabular}{cc|ccc} \hline
         $\bm{h}(x_t,x_{t-1})$& $K$ & $\log{L_{PLE}(\bm{\hat{\theta}}_{PLE})}$&AIC&  PIC\\ \hline
         $x_t \otimes x_{t-1}$ & 4 &-87896.88&175801.7& 175846.1\\
         $(x_t \otimes x_{t-1},x_t \otimes x_{t-1}^2)$ & 8&-87541.25 &175098.5 & 175187.2\\
         $(x_t \otimes x_{t-1},x_t^2 \otimes x_{t-1})$ & 8&-87527.17 & 175070.3 & 175159.0\\
         $(x_t \otimes x_{t-1},x_t \otimes x_{t-1}^2, x_t^2 \otimes x_{t-1}, x_t^2 \otimes x_{t-1}^2)$ & 16&-86390.15 & \bf{172812.3} & \bf{172989.2}\\ \hline
    \end{tabular}\\
    \vspace{5mm}
    \begin{tabular}{cc|ccc} \hline
         $\bm{h}(x_t,x_{t-1})$& $K$ & $\log{L_{PLE}(\bm{\hat{\theta}}_{PLE})}$&AIC&  PIC\\ \hline
         $x_t \otimes x_{t-1}$ & 4 & -79582.75& 159173.50& 159217.9\\
         $(x_t \otimes x_{t-1},x_t \otimes x_{t-1}^2)$ & 8&-79456.37&158920.74& 159017.4\\
         $(x_t \otimes x_{t-1},x_t^2 \otimes x_{t-1})$ & 8&-79477.20 & 158962.40 & 159059.1\\
         $(x_t \otimes x_{t-1},x_t \otimes x_{t-1}^2, x_t^2 \otimes x_{t-1}, x_t^2 \otimes x_{t-1}^2)$ & 16&-76010.07 & \bf{152036.14} & \bf{152124.8}\\ \hline
    \end{tabular}
    \caption{PLE on bivariate LFP-spike dataset (top: channel = 2, bottom: channel = 75). The first row is equivalent to Gaussian VAR(1).}
    \label{tab:PLE-bivariate-LFP}
\end{table}

\section{Concluding remarks}

We proposed the minimum information Markov model by defining the dependence function and the stationary distribution. This separation is analogous to the copula modeling where the joint distribution is decomposed into a copula function and marginals. Indeed, in our modeling, the parameter in the dependence function and the stationary distribution are orthogonal with respect to the divergence rate, a KL-divergence for Markov kernels.
Therefore, parameter estimation based on the likelihood can be efficiently performed using MCLE or PLE. For higher order, larger dimension, and greater sample size, PLE is more practical in terms of computational time because it reduces to a standard logistic regression problem, which can be accelerated by various methods, including learning rate adjustment, stochastic gradient descent, and subsampling, as confirmed by our simulation studies.

In this paper, we considered the Gaussian autoregressive models primarily, which are completely tractable in its explicit expression. The theoretical understanding of infinite-state Markov chains is not yet fully developed. For example, the existence of the minimum information Markov model is not guaranteed in general when the dependence and the stationary distribution are both specified. 

Another important question in the minimum-information Markov model concerns the choice of the number and form of the dependence functions. In our experiments, we employed simple low-degree polynomials consisting of a single Kronecker product as an initial approach. In principle, the dependence functions can be complex or nonlinear; however, this often makes their interpretation less intuitive.

\section*{Acknowledgements}

I.S. is supported by RIKEN Junior Research Associate program.
T.S. is supported by JSPS KAKENHI Grant Numbers JP19K11865 and JP21K11781.
ChatGPT (version 5.2) was used solely for English proofreading and stylistic improvements. The authors are responsible for the scientific content of this work.

\appendix

\bibliographystyle{plain}
\bibliography{ref}
\section{Proof of Proposition~\ref{prop:ar1}} \label{appendix:ar1}

We employ the notations from Proposition~\ref{prop:pythagorean} where $p=1$ and $r = \mathcal{N}(0,\tau^2)$. 
The proof is briefly composed of three steps: (i) construct a Markov kernel that belongs to $\mathcal{E}$, (ii) formulate a divergence minimization problem where the optimal solution should lie in $\mathcal{E} \cap \mathcal{M}$, (iii) show this problem is closed convex to guarantee the (unique) existence of its optimal solution. 

We begin by confirming that a Markov kernel with the same dependence structure with Gaussian AR(1) is constructable. 

\begin{lemma} \label{lemma:ar1-markov-kernel}
For $\theta \in \mathbb{R}$ and $D \in\mathbb{R}$ such that $D > |\theta|$, there exists $K \in\mathbb{R}$ (and $C \in\mathbb{R}$) such that
$$p(y|x) = \exp{(\theta xy + Ky^2 - Kx^2 - Dy^2 -C)}$$
is a stationary Markov kernel.
Moreover, its stationary distribution is $\mathcal{N}(0,1/2\sqrt{D^2-\theta^2})$. $K$ and $C$ depend on $D$ and $\theta$. 
\end{lemma}
\begin{proof}
For all $x \in \mathbb{R}$, $p(y|x)$ should satisfy
$\int_{-\infty}^\infty \exp{(\theta xy + Ky^2 - Kx^2 - Dy^2 -C)} \mathrm{d}y = 1$. By solving it via Gaussian integral, we have
$$\sqrt{\frac{\pi}{-(K-D)}}\exp{\left(-Kx^2 - \frac{(\theta x)^2}{4(K-D)}\right)} = \exp{(C)}$$
which leads to 
$K = \frac{1}{2}(D \pm \sqrt{D^2 - \theta^2})$.
This $K$ is indeed strictly smaller than $D$ regardless of the sign in the $\pm$.

Assume the stationary distribution is $\mathcal{N}(0,\tau^2)$. Then, $K$, $D$, and $\theta$ should satisfy
$$\int p(y|x)\mathcal{N}(x|0,\tau^2) dx = \mathcal{N}(y|0,\tau^2) $$
$$\frac{\theta^2}{\frac{1}{\tau^2}+2K} + 2(K - D) = -\frac{1}{\tau^2}$$
$$K = \frac{1}{2}\left(D-\frac{1}{\tau^2} \pm \sqrt{D^2-\theta^2}\right)$$
\noindent This $K$ is indeed strictly smaller than $D$ regardless of the sign in the $\pm$.

Therefore, $K$ is uniquely specified as $K = \frac{1}{2}(D - \sqrt{D^2 - \theta^2})$.
The variance of the stationary distribution is obtained by solving
$$\frac{1}{2}(D - \sqrt{D^2 - \theta^2}) = \frac{1}{2}(D - \frac{1}{\tau^2} + \sqrt{D^2 - \theta^2}).$$
\end{proof}

By using the constructed Markov kernel $q \in \mathcal{E}$, we show that the divergence minimization problem reduces to a closed convex problem. Consider
$$p(y,x) = \mathcal{N}(\bm{0}, M) \in \mathcal{M},\ M = \begin{pmatrix}
    \tau^2&m\\
    m&\tau^2
\end{pmatrix}$$
and 
$$q(y|x) = \exp{(\theta xy + Ky^2 - Kx^2 - Dy^2 -C)} \in \mathcal{E},$$
where $M \succ O$. Then, the divergence rate is explicitly calculated as
\begin{align*}
D(p|q)&= \int p(y,x) \log{p(y|x)} dxdy - \int p(y,x) \log{q(y|x)} dxdy\\
&= -1 + \frac{1}{2} -\frac{1}{2} \log{\{(2\pi)^2|M|\}} + \frac{1}{2}\log{(2\pi\tau^2)} - \mathrm{Tr}(QM) + C\\ 
\end{align*}
where $Q = \begin{pmatrix}
    K-D&\theta/2\\
    \theta/2 & -K
\end{pmatrix}$. Therefore, the problem is reformulated as a minimization problem with respect to $m \in \mathbb{R}$:
$$\mathrm{minimize}\ D(m)$$
$$\mathrm{s.t.}\ M \succ O, m \in \mathbb{R}.$$
The positive definiteness in the constraint can be relaxed to positive semidefiniteness since the objective function tends to positive infinity for $m$ such that $M$ admits a zero eigenvalue.
Then, the existence and the uniqueness of the optimal solution are guaranteed since the objective function is convex with respect to $M$, which is an affine function of $m$, and the feasible region is a closed convex set.

\section{Extension of Proposition~\ref{prop:ar1} to $d \geq 2$ and $p \geq 2$} \label{appendix:var1-and-var2}

\subsection{Case: $d=1$}

We again employ the notations $\mathcal{E}$ and $\mathcal{M}$ from Proposition~\ref{prop:pythagorean}, adapted to each context. 
Naturally, the argument presented in \ref{appendix:ar1} extends to higher-dimensional settings.
\begin{proposition} \label{prop:var1}
    Suppose $d = 1$ and $p \geq 1$. 
    For the dependence function $\bm{h}(x_t,x_{t-1}) = x_t \otimes x_{t-1}$ and the stationary distribution $r(x) = \mathcal{N}(\bm{0}, B), B \succ O $, the minimum information Markov model \eqref{eq:mininfo-markov-model} exists uniquely. Moreover, it is identical to the Gaussian VAR(1) model (Example~\ref{example:var1}).
\end{proposition}

The following lemmas are the higher-dimensional version ($d=2, p\geq 2$) of Lemma~\ref{lemma:ar1-markov-kernel}.

\begin{lemma}
    For any $\Theta \in \mathbb{R}^{p\times p}$, there exist $D,K\in\mathbb{R}^{p\times p}$ and $C\in\mathbb{R}$ such that
    $$p(\bm{y}|\bm{x}) = \exp{\left(\bm{x}^\top \Theta \bm{y}  + \bm{y}^\top K\bm{y} - \bm{x}^\top K\bm{x} - \bm{y}^\top D\bm{y} - C \right)}, \bm{x}, \bm{y} \in \mathbb{R}^p$$
    is a stationary Markov kernel.
\end{lemma}
\begin{proof}
    For all $x$, the Markov kernel $p$ should satisfy
    \begin{align*}
    1 &= \int p(\bm{y}|\bm{x}) d\bm{y}\\
    &= \exp{(- \bm{x}^\top K\bm{x}-C)} \int \exp{\left(-\frac{1}{2}\bm{y}^\top \{-2(K-D)\}\bm{y} + \bm{x}^\top \Theta \bm{y} \right)} d\bm{y}\\
    &\propto \exp{\left(-\frac{1}{4} \bm{x}^\top \Theta (K-D)^{-1} \Theta^\top \bm{x} - \bm{x}^\top K \bm{x} \right)},
    \end{align*}
    which leads to 
    \begin{equation} 
    \Theta (K-D)^{-1} \Theta^\top + 4K = O.
    \end{equation}
    Let $\alpha>0$ and suppose that $D=\alpha I$ and $R=D-K=\alpha I-K$.
    Then, we have
    $$-\Theta R^{-1}\Theta^\top + 4(\alpha I-R) = O,$$
    or equivalently
    $$\frac{1}{4}\Theta R^{-1}\Theta^\top + R = \alpha I.$$
    Normalize $R=\alpha S$ to obtain
    $$S + \frac{1}{4\alpha^2}\Theta S^{-1}\Theta^\top = I.$$
    This equation has a solution around $S=I$ if $\alpha$ is sufficiently large.
    Indeed, if $T=\Theta/(2\alpha)$ is sufficiently small,
    \[
     S +  TS^{-1}T^\top = I
    \]
    has a solution $S$ around $S=I$ due to the implicit function theorem
    because the derivative of $S+TS^{-1}T^\top$ with respect to $S$ at $(S,T)=(I,O)$ is the identity map on $\mathbb{R}^{p\times p}$.
    
    Tracing back the above calculation, we have
    $$p(\bm{y}|\bm{x}) \propto \exp\left(-\frac{1}{4}(\bm{y}-\frac{S^{-1}\Theta^\top}{2\alpha}\bm{x})^\top S(\bm{y}-\frac{S^{-1}\Theta^\top}{2\alpha}\bm{x})\right).$$
    This Markov kernel is stationary if $\alpha$ is sufficiently large (because $\Theta$ is fixed and $S$ is around $I$).
\end{proof}

\begin{lemma}
    A Markov model with a kernel
    $$p(\bm{y}|\bm{x}) = \exp{\left(\bm{x}^\top \Theta \bm{y} + \bm{y}^\top K\bm{y} - \bm{x}^\top K\bm{x} - \bm{y}^\top D\bm{y} - C \right)}, \bm{x}, \bm{y} \in \mathbb{R}^p$$
    has a Gaussian distribution $\mathcal{N}(\bm{0},B)$ as its stationary distribution when $(K, \Theta) \in \mathbb{R}^{p\times p} \times \mathbb{R}^{p\times p}$ satisfies
    $$2(D-K) -\Theta^\top (B^{-1} + 2K)^{-1} \Theta   = B^{-1}.$$
\end{lemma}
\begin{proof}
Assume the stationary distribution is $\mathcal{N}(0,B)$. The condition that $B$ should satisfy is
$$\int p(\bm{y}|\bm{x}) \mathcal{N}(\bm{x}; 0,B) d\bm{x} = \mathcal{N}(\bm{y}; 0,B)$$
The left hand side is proportional to 
$\exp{\left(-\frac{1}{2}\bm{y}^\top (-2(K-D)-\Theta^\top (B^{-1} + 2K)^{-1} \Theta)\bm{y}\right)}$.
\end{proof}

\begin{corollary}
    There exist $D \in \mathbb{R}^{p\times p}$, $K \in \mathbb{R}^{p\times p}$ and $C \in \mathbb{R}$ such that Lemma 5 holds.
\end{corollary}
\begin{proof}
    Let $\alpha > 0$ and suppose that $D = \alpha I$ and $R = D-K = \alpha I - K$. 
    Then, we have 
    $$2R - \Theta^\top (B^{-1} + 2(\alpha I - R))^{-1}\Theta = B^{-1}.$$
    Denote $R' = B^{-1} - 2R + 2\alpha I$ to obtain
    $$-\Theta^\top R'^{-1} \Theta = R'-2\alpha I$$
    Normalize $R' = 2\alpha S$ to obtain
    $$S + \frac{1}{4\alpha^2}\Theta^\top S^{-1} \Theta = I.$$
    This is the same form in Lemma 4, thus it has a positive definite solution $S$ around $S = I$ if $T = \Theta/(2\alpha)$ is sufficiently small. Its positive definiteness is guaranteed because the minimum eigenvalue of the solution map $S(T)$ is continuous and $\lim_{T \to O} \lambda_{\min}(S(T)) = \lambda_{\min}(I) = 1$ holds.
    Tracing back the above calculation, we have 
    $$B^{-1} = 2\alpha S + 2(\alpha I-K) - 2\alpha I  = 2\alpha S - 2K.$$
    Such $B$ exists as a positive definite matrix if $\alpha$ is sufficiently large since
    $$\lambda_{\min}(B^{-1}) \geq 2\alpha\lambda_{\min}(S) - 2 \|K\| > 0$$
    where $\|\cdot\|$ denotes the spectral norm of a matrix.
\end{proof}

For a higher dimensional minimum information Markov model which is related to the Gaussian VAR(1), the following minimization problem shall be considered.
Suppose $(y, x)$ follows a Gaussian distribution with marginal covariance $B \ (\succ O)$:
\[
p(y, x) = \mathcal{N}\!\left(
    \bm{0}, M
\right),
\]
where
$M = 
\begin{pmatrix}
    B & M_1 \\
    M_1^\top & B
\end{pmatrix} \succ O$
and
$$q(y|x) = \exp{\left(x^\top \Theta y  + y^\top K y - x^\top K x - y^\top D y - C\right)} \in \mathcal{E}. $$
Then, the divergence rate between $p(y|x)$ and $q(y|x)$ is explicitly calculated as 
\begin{align*}
D(p|q) = D(M) &= \int\int p(y,x) \log\frac{p(y|x)}{q(y|x)} dxdy\\
&= -\frac{2p}{2} + \frac{p}{2} - \frac{1}{2}\log{\left((2\pi)^{2p}|M|\right)} + \frac{1}{2}\log{((2\pi)^p |B|)} - \mathrm{Tr}(\begin{pmatrix}
    K-D&\Theta^\top/2\\
    \Theta/2&-K
\end{pmatrix}M) + C.
\end{align*}

\noindent Therefore, the problem is reformulated as a minimization problem on matrices:
$$\mathrm{minimize}\ D(M_1)$$
$$\mathrm{s.t.}\ M \succ O, M_1\in\mathbb{R}^{p\times p}.$$
The positive definiteness in the constraint can be relaxed to positive semidefiniteness since the objective function tends to positive infinity for $M_1$ such that $M$ admits a zero eigenvalue.
Then, the existence and the uniqueness of the optimal solution are guaranteed since the objective function is convex with respect to $M_1$ and the feasible region is a closed convex set.

\subsection{General case}

We first consider the case of $d=2$ as an example, but the similar argument holds for a general $d \in \mathbb{N}$.
Assume the stationary distribution of $p(x_3|x_2,x_1)$ is Gaussian $\mathcal{N}(0,B)$, where $B$ is a positive definite $p \times p$ matrix here. Then, the joint distribution is also Gaussian, i.e., $p(x_3,x_2,x_1) = \mathcal{N}\left(\bm{0}, M^{(3)}\right) \in \mathcal{M}$ for some matrices $(M_1, M_2) \in \mathcal{D}$ where
$$M^{(3)} := \begin{pmatrix}
    B&M_1&M_2\\
    M_1^\top& B & M_1\\
    M_2^\top & M_1^\top & B
\end{pmatrix}, M^{(2)} := \begin{pmatrix}
    B&M_1\\
    M_1^\top& B\\
\end{pmatrix},
$$
$$\mathcal{D}:=\Big\{(M_1,M_2)\ \big|\ M^{(2)} \succ O, M^{(3)}\succ O\Big\} =\Big\{(M_1,M_2)\ \big|\ M^{(3)}\succ O\Big\}.$$

Also, consider another Gaussian kernel $q$ in the following form:
$$\log{q(x_3|x_2,x_1)} = (x_3,x_2,x_1) Q (x_3,x_2,x_1)^\top - c, \quad Q \in \mathbb{R}^{3p \times 3p}, c\in\mathbb{R}.$$
The divergence rate from $p$ to $q$ is 
\begin{align*}
D(p|q) = D(M_1,M_2) &= \int p(x_3,x_2,x_1) \log\frac{p(x_3|x_2,x_1)}{q(x_3|x_2,x_1)} dx_3dx_2dx_1\\
&= \left(-\frac{3}{2}p + \frac{2}{2}p - \log{\sqrt{(2\pi)^{3p}|M^{(3)}|}} + \log{\sqrt{(2\pi)^{2p} |M^{(2)}|}}\right) - \left( \mathrm{Tr}(QM^{(3)}) - c \right)\\
&= -\frac{1}{2}\log\det M^{(3)} + \frac{1}{2}\log\det M^{(2)} - \mathrm{Tr}(QM^{(3)}) + \mathrm{const.}\\
%
\end{align*}
which is obviously a strictly convex function of $M_2$. Also, inside $\mathcal{D}$ where the Schur complement is valid, 
\begin{align*}
    D(p|q) &= -\frac{1}{2}\log\det(B-\begin{pmatrix}
    M_1&M_2
\end{pmatrix} \begin{pmatrix}
    B&M_1\\
    M_1^\top& B
\end{pmatrix}^{-1} \begin{pmatrix}
    M_1^\top\\M_2^\top
\end{pmatrix}) - \mathrm{Tr}(QM) + \mathrm{const.}, \\ 
\end{align*}
the mapping $\left(\begin{pmatrix}
    B&M_1\\
    M_1^\top& B
\end{pmatrix}, \begin{pmatrix}
    M_1^\top\\M_2^\top
\end{pmatrix}\right) \to \begin{pmatrix}
    M_1&M_2
\end{pmatrix} \begin{pmatrix}
    B&M_1\\
    M_1^\top& B
\end{pmatrix}^{-1} \begin{pmatrix}
    M_1^\top\\M_2^\top
\end{pmatrix}$ is matrix convex because its epigraph 
\begin{align*}
\mathcal{E} &= \{\left(\begin{pmatrix}
    B&M_1\\
    M_1^\top& B
\end{pmatrix}, \begin{pmatrix}
    M_1^\top\\M_2^\top
\end{pmatrix},T\right)\ |\ T \succeq \begin{pmatrix}
    M_1&M_2
\end{pmatrix} \begin{pmatrix}
    B&M_1\\
    M_1^\top& B
\end{pmatrix}^{-1} \begin{pmatrix}
    M_1^\top\\M_2^\top
\end{pmatrix}, \begin{pmatrix}
    B&M_1\\
    M_1^\top& B
\end{pmatrix} \succ O\} \\
&= \{\left(\begin{pmatrix}
    B&M_1\\
    M_1^\top& B
\end{pmatrix}, \begin{pmatrix}
    M_1^\top\\M_2^\top
\end{pmatrix},T\right)\ |\begin{pmatrix}
    T&M_1&M_2\\
    M_1^\top& B & M_1\\
    M_2^\top & M_1^\top & B
\end{pmatrix} \succeq O, \begin{pmatrix}
    B&M_1\\
    M_1^\top& B
\end{pmatrix} \succ O\} 
\end{align*}
is a convex set, and $-\log\det$ is both strictly convex and monotonically decreasing with respect to Loewner order, hence $D(p|q)$ is also a strictly convex function of $M_1$.
Therefore, to guarantee the existence and uniqueness of an optimal solution to the divergence minimization problem, since $D$ is continuous on $\mathcal{D}$, it suffices to show that $D$ blows up on the boundary of $\mathcal{D}$.

To do so, it is equivalent to show that if $D(p|q)$ is bounded, then $p$ belongs to a compact set of $\mathcal{D}$. 

\begin{lemma}
Denote a fixed $d$-th order stationary Gaussian kernel as $q(x_{d+1}|x_{d:1}) \propto \exp{(x_{d+1:1}^\top Q^{(d+1)} x_{d+1:1})}$, where $Q^{(d+1)} \in \mathbb{R}^{p(d+1)\times p(d+1)}$ is a fixed negative definite matrix.
Denote another $d$-th order stationary Gaussian kernel as $p$ such that its stationary distribution $\mathcal{N}(\bm{0},B)$ and joint distribution $p(x_{d+1:1})$ is  $ \mathcal{N}(0,M^{(d+1)})$, where $M^{(d+1)}$ is a $p(d+1) \times p(d+1)$ block Toeplitz matrix. 
Define the set of such $p$ as 
$$\mathcal{A} := \{M^{(d+1)} \mid p:\mathrm{stationary}, (M^{(d+1)})_{11} = B, D(p|q) \leq K\}, K:\mathrm{constant}$$
Then, $\mathcal{A}\ (\subset \mathbb{S}^{p(d+1)}_{++})$ is compact.
\end{lemma}

\begin{proof}
The divergence rate can be decomposed as
\begin{align*}
    D(p(x_{d+1:1})|q(x_{d+1:1})) &= \int p(x_{d+1:1})\log{\frac{p(x_{d+1} | x_{d:1})}{q(x_{d+1} | x_{d:1})}} dx_{d+1}\dots dx_1\\
    &= \int p(x_{d+1:1})\log{\frac{p(x_{d+1} | x_{d:1})}{p(x_{d+1} | x_{d:2})}} dx_{d+1}\dots dx_1 - \int p(x_{d+1:1})\log{\frac{q(x_{d+1} | x_{d:1})}{q(x_{d+1} | x_{d:2})}} dx_{d+1}\dots dx_1\\
    &+ \int p(x_{d+1:2})\log{\frac{p(x_{d+1} | x_{d:2})}{q(x_{d+1} | x_{d:2})}} dx_{d+1}\dots dx_1\\
\end{align*}
The first term is non-negative. The second term
\begin{align} \label{eq:second-term}
    \int p(x_{d+1:1})\log{\frac{q(x_{d+1} | x_{d:1})}{q(x_{d+1} | x_{d:2})}} dx_{d+1}\dots dx_1 &= (\mathrm{Tr}(Q^{(d+1)}M^{(d+1)}) - c^{(d+1)}) - (\mathrm{Tr}(Q^{(d)}M^{(d)}) - c^{(d)}),
\end{align}
where $c^{(d)}$ and $c^{(d+1)}$ are normalizing constants, is evaluated as follows. Using $|\operatorname{Tr}(AB)|\le \|A\|_{\mathrm{F}}\|B\|_{\mathrm{F}}$, we derive the
compact upper bound
\begin{align*}
|\mathrm{Tr}(Q^{(d+1)}M^{(d+1)}) - \mathrm{Tr}(Q^{(d)}M^{(d)})| &= \operatorname{Tr}\!\left(Q^{(d+1)}_{d+1,d+1} \, M^{(d+1)}_{d+1,d+1}\right)
+\sum_{k=1}^{d}
\operatorname{Tr}\!\left(
   Q^{(d+1)}_{k,d+1} \, M^{(d+1)}_{k,d+1}
 + Q^{(d+1)}_{d+1,k} \, M^{(d+1)}_{d+1,k}
\right)\\
&\;\le\;
\|Q^{(d+1)}_{d+1,d+1}\|_{\mathrm{F}}
 \,\|M^{(d+1)}_{d+1,d+1}\|_{\mathrm{F}}
+
\sum_{k=1}^{d}
\Bigl(
  \|Q^{(d+1)}_{k,d+1}\|_{\mathrm{F}}
   \,\|M^{(d+1)}_{k,d+1}\|_{\mathrm{F}}
  +
  \|Q^{(d+1)}_{d+1,k}\|_{\mathrm{F}}
   \,\|M^{(d+1)}_{d+1,k}\|_{\mathrm{F}}
\Bigr).
\end{align*}
Since $M^{(d+1)}$ is positive definite, the
block Cauchy--Schwarz inequality yields
\[
\|M^{(d+1)}_{k,d+1}\|
\;\le\;
\sqrt{\,
  \|M^{(d+1)}_{k,k}\|\;
  \|M^{(d+1)}_{d+1,d+1}\|
}\qquad (1\le k\le d),
\]
for any consistent operator norm $\|\cdot\|$. Therefore, 
{\makeatletter
\fontsize{0.8\dimexpr\f@size pt\relax}{1.1\dimexpr\f@size pt\relax}\selectfont
\begin{align*}
|\mathrm{Tr}(Q^{(d+1)}M^{(d+1)}) - \mathrm{Tr}(Q^{(d)}M^{(d)})| 
&\;\le\;
\|Q^{(d+1)}_{d+1,d+1}\|_{\mathrm{F}}
 \,\|M^{(d+1)}_{d+1,d+1}\|_{\mathrm{F}}
+
\sqrt{\|M^{(d+1)}_{d+1,d+1}\|}
\sum_{k=1}^{d}
\Bigl(
  \|Q^{(d+1)}_{k,d+1}\|_{\mathrm{F}}
   \sqrt{\|M^{(d+1)}_{k,k}\|}
  +
  \|Q^{(d+1)}_{d+1,k}\|_{\mathrm{F}}
   \sqrt{\|M^{(d+1)}_{k,k}\|}
\Bigr)\\
&= \|Q^{(d+1)}_{d+1,d+1}\|_{\mathrm{F}}
 \,\|B\|_{\mathrm{F}}
+
\sqrt{\|B\|}
\sum_{k=1}^{d}
\Bigl(
  \|Q^{(d+1)}_{k,d+1}\|_{\mathrm{F}}
   \sqrt{\|B\|}
  +
  \|Q^{(d+1)}_{d+1,k}\|_{\mathrm{F}}
   \sqrt{\|B\|}
\Bigr)\\
&< \infty
\end{align*}
}

Hence, the second term~\eqref{eq:second-term} is bounded. 
The third term is equal to $D(p(x_{d:1})|q(x_{d:1}))$ by its definition. 
By induction, there exists $\tilde{K}\ (<\infty)$ such that for every $j=1,\dots, d+1$, 
$$D(p(x_{j:1})|q(x_{j:1})) < \tilde{K}.$$
This implies that the KL divergence between the joint densities $p(x_{d+1:1})$ and $q(x_{d+1:1})$ is also bounded:
$$\int p(x_{d+1:1})\log{\frac{p(x_{d+1:1})}{q(x_{d+1:1})}} dx_{d+1}\dots dx_1 = \sum_{j=1}^{d+1} D(p(x_{j:1})|p(x_{j:1})) \leq (d+1) \tilde{K} < \infty,$$
hence $\mathcal{A}$ is a compact set in $\mathbb{S}^{p(d+1)}_{++}$.
\end{proof}

In conclusion, in these Gaussian-VAR-related cases, the divergence rate minimization problem reduces to an optimization problem on a finite-dimensional matrix space, and the existence of the optimal solution is guaranteed. The uniqueness also follows from Proposition~\ref{prop:pythagorean}. However, the existence of the optimal solution for the aforementioned divergence rate minimization problem is not trivial in general cases.
\end{document}